\newcommand{\includeappendix}{} 
\documentclass[10pt]{article}

\usepackage[paperwidth=6.75in, paperheight=10in, margin=1in]{geometry}
\usepackage{cite}
\usepackage{savetrees}
\usepackage{enumitem}
\usepackage{algorithm,algpseudocode}
\usepackage{tabularx}
\usepackage{amsfonts,amsmath,amsthm,stmaryrd}

\newtheorem{Thm}{Theorem}[section]
\newtheorem{Lem}[Thm]{Lemma}
\newtheorem{Prop}[Thm]{Proposition}

\newtheorem{Def}[Thm]{Definition}
\newtheorem{Exm}[Thm]{Example}

\def\R{\mathbb{R}}

\def\N{\mathbb{N}}

\def\Z{\mathbb{Z}}

\def\M{\mathcal{M}}
\def\D{\mathcal{D}}

\def\Y{\mathcal{Y}}
\usepackage{bbm}
\def\1{\mathbbm{1}}
\newcommand{\prob}[1]{\Pr\left[\begin{array}{@{}c@{}}#1\end{array}\right]}
\renewcommand{\epsilon}{\varepsilon}

\usepackage[
	n,
	operators,
	advantage,
	sets,
	adversary,
	landau,
	probability,
	notions,	
	logic,
	ff,
	mm,
	primitives,
	events,
	complexity,
	oracles,
	asymptotics,
	keys]{cryptocode}
\usepackage{adjustbox}

\newcommand{\round}[1]{\left\lfloor#1 \right\rceil}

\renewcommand{\sf}[1]{\ifmmode\mathsf{#1}\else\textsf{#1}\fi}
\newcommand{\adbf}[1]{\ifmmode\mathbf{#1}\else\textbf{#1}\fi}
\newcommand{\bs}[1]{\boldsymbol{#1}}

\def\1{\mathbbm{1}}

\usepackage{multirow}

\usepackage{booktabs}

\newcommand{\autofit}[1]{\maxsizebox{\linewidth}{!}{\begin{varwidth}{10\linewidth}#1\end{varwidth}}}

\setlist{nosep}
\usepackage{pifont}

\makeatletter
\newcommand{\appref}[1]{%
    \ifdefined\includeappendix
        \ifcsdef{r@#1}{%
            Appendix~\ref{#1}%
        }{%
            \textbf{\color{red}{\Huge APPENDIX NOT FOUND}}%
        }%
    \else
        the appendix%
    \fi
}
\makeatother
\usepackage{xurl}
\usepackage[colorlinks]{hyperref}

\title{GPM: The Gaussian Pancake Mechanism for Planting Undetectable Backdoors in Differential Privacy}

\date{}

\author{
    Haochen Sun and Xi He \\
    Cheriton School of Computer Science, University of Waterloo \\
    \texttt{\{haochen.sun, xi.he\}@uwaterloo.ca}
}

\begin{document}

\maketitle

\begin{abstract}
    Differential privacy (DP) has become the gold standard for preserving individual privacy in data analysis. However, an implicit yet fundamental assumption underlying these rigorous privacy guarantees is the correct implementation and execution of DP mechanisms. Several incidents of unintended privacy loss have occurred due to numerical issues and inappropriate configurations of DP software, which have been successfully exploited in privacy attacks. To better understand the seriousness of defective DP software, we ask the following question: is it possible to elevate these passive defects into active privacy attacks while maintaining covertness?

    To address this question, we present the \emph{Gaussian pancake mechanism (GPM)}, a novel mechanism that is computationally indistinguishable from the widely used Gaussian mechanism (GM), yet exhibits arbitrarily weaker statistical DP guarantees. This unprecedented separation enables a new class of backdoor attacks: by indistinguishably passing off as the authentic GM, GPM can covertly degrade statistical privacy. Unlike the unintentional privacy loss caused by GM's numerical issues, GPM is an adversarial yet undetectable backdoor attack against data privacy. We formally prove GPM's covertness, characterize its statistical leakage, and demonstrate a concrete distinguishing attack that can achieve near-perfect success rates under suitable parameter choices, both theoretically and empirically.

    Our results underscore the importance of using transparent, open-source DP libraries and highlight the need for rigorous scrutiny and formal verification of DP implementations to prevent subtle, undetectable privacy compromises in real-world systems.
\end{abstract}

\section{Introduction}

In the modern digital landscape, an unprecedented volume of data is continuously collected, aggregated, and analyzed across diverse domains, including healthcare, finance, social networks, and cloud-based services. As data-driven technologies proliferate, ensuring the privacy of individuals represented in these datasets has become a critical and widely acknowledged challenge. Differential privacy (DP)~\cite{DBLP:conf/tcc/DworkMNS06, DBLP:journals/fttcs/DworkR14} has emerged as the de facto gold standard for formal privacy guarantees. It offers strong theoretical protection by ensuring that the presence or absence of a single individual's data has a provably bounded influence on the output of a computation, thereby statistically limiting the risk of privacy breaches against any adversary.

However, realizing the strong theoretical guarantees offered by differential privacy in practice depends on a crucial yet often overlooked assumption: that the mechanisms are correctly implemented and faithfully executed. In other words, the software systems delivering differential privacy must precisely adhere to the formal design of the mechanisms. Any deviation, whether due to programming errors, numerical instability, or misinterpretation of the specification, can undermine the intended protections and leave users vulnerable to privacy breaches.

Indeed, several incidents have demonstrated the risks posed by compromised or improperly implemented differential privacy mechanisms. In particular, Laplace, Exponential, and Gaussian mechanisms have all been found vulnerable to numerical issues~\cite{DBLP:conf/ccs/Mironov12,DBLP:conf/ccs/Ilvento20,DBLP:conf/sp/JinMRO22}. This is because the original definitions and guarantees of differential privacy are formulated in the real number space but must ultimately be realized using finite-precision floating-point arithmetic. Such discrepancies have opened the door to distinguishing attacks that exploit these vulnerabilities. In particular, by checking only the mechanism output, the attacker identifies the correct input database from a pair of neighbouring databases with a success rate of over 92\% on widely used algorithms such as DP-SGD~\cite{DBLP:conf/sp/JinMRO22}. These findings have prompted the release of security patches in major differential privacy libraries~\cite{DBLP:conf/ccs/HolohanB024}. Beyond numerical issues, practical misconfigurations have also raised concerns. For example, in 2017, it was revealed that earlier versions of macOS employed local differential privacy for data collection but used excessively large privacy budgets, rendering the guarantees effectively meaningless~\cite{DBLP:journals/corr/abs-1709-02753}.

The existence of compromised DP software raises two natural and urgent questions. First, given that distinguishing attacks can exploit passive implementation issues, can such compromised implementations be intentionally elevated into \textit{active} privacy attacks that an adversary engineers and uses? Second, how \textit{covert} can these deviations be? In particular, intellectual-property concerns may keep deployed DP implementations closed-source. However, the aforementioned vulnerabilities resulting from numerical issues and excessive privacy budgets are still detectable by examining only the output distribution. But do such deviations exist, whether introduced deliberately or accidentally, that make the outputs computationally indistinguishable from the genuine outputs while leaving no obvious forensic traces (e.g., unexpected system calls and extra network traffic)?

If such deviations exist, they can be weaponized as backdoors, which substantially and silently undermine the formal privacy guarantees, either theoretically or empirically. A backdoor attacker could embed a compromised implementation into a closed-source DP library or service so that, to all ordinary examinations and black-box tests, or to the server executing the mechanisms itself, the mechanism appears correct and privacy-preserving. Unless a full white-box inspection of the implementation is performed (which is infeasible for many proprietary systems), neither servers that rely on the library nor downstream users can detect the backdoor or realize that the stated DP guarantees have been violated. The backdoor would therefore elevate from accidental and passive leakage (e.g., rounding or floating-point artifacts) to intentional, adversarial insertion of covert vulnerabilities that remain hidden to standard testing and auditing practices.

To investigate this threat, we introduce the \textbf{Gaussian pancake mechanism (GPM)}, an indistinguishable backdoor attack against data privacy, parameterized by a secret backdoor key. To any party not knowing the backdoor key, GPM is computationally indistinguishable from the widely used Gaussian mechanism (GM). That is, no polynomial-time adversary (including the server and the downstream users) can efficiently tell their outputs apart and thus discover the existence of the backdoor. Critically, however, under certain parameter regimes, GPM exhibits substantially weaker privacy guarantees: to the backdoor attacker knowing the secret backdoor key, arbitrarily large privacy leakage could be induced. Given that GM is ubiquitous in deployed systems, the existence of constructions like GPM leads to a troubling conclusion: \textbf{any non-open-source DP software is not safe for production use, even with auditing in place}. In particular, passing black-box audits does not rule out the existence of undetectable backdoors in the software that violate the claimed privacy guarantees.

The existence of GPM has significant implications for the design and deployment of privacy-preserving systems. Most importantly, it highlights the necessity of using transparent, open-source implementations such as OpenDP~\cite{opendp}, Opacus~\cite{opacus}, and Diffprivlib~\cite{diffprivlib}. Without access to and scrutiny of implementation details, backdoors like GPM could be seamlessly integrated and remain undetected, while appearing indistinguishable from authentic mechanisms. Moreover, further verification measures may be adopted to mitigate this risk. These include programming-language-based formal verification of DP libraries~\cite{DBLP:conf/ccs/LoknaPDV23, DBLP:conf/vmcai/ReshefKSD24, DBLP:conf/ccs/WangDKZ20, DBLP:journals/corr/abs-1909-02481}, which can ensure semantic correctness from specification to implementation. Additionally, cryptographic techniques such as zero-knowledge proofs and verifiable computation~\cite{vfuzz,dprio,DBLP:conf/iclr/ShamsabadiTCBHP24,KCY21,BC23,vddp} can be used to verify that a mechanism was executed correctly without requiring trust in the server, which is particularly relevant in modern contexts as DP is increasingly adopted in security-critical applications.

Our contributions can be summarized as follows:

\begin{itemize}[leftmargin=*]
    \item We formally define the Gaussian pancake mechanism (GPM) and describe an indistinguishable backdoor attack based on this construction. \textbf{(Section~\ref{sec:setup})}
    
    \item We analyze the statistical privacy guarantees of GPM and show that, despite satisfying the same computational DP guarantees as the Gaussian mechanism (GM), GPM may require an arbitrarily large lower bound on $\epsilon$ to satisfy $\left(\epsilon, \delta=0.5\right)$-DP\footnote{$\epsilon$ increases as $\delta$ decreases.}, rendering it unsuitable for practical deployment. \textbf{(Section~\ref{sec:privacy})}
    
    \item We study the distinguishing attack, a concrete attack that exploits the additional privacy leakage introduced by the GPM backdoor, prove that its success rate can be arbitrarily close to 1, and discuss the potential mitigations against it. \textbf{(Sections~\ref{sec:da} and \ref{sec:mit})}
    
    \item We provide experimental evaluations of the GPM backdoor's effectiveness, demonstrating that distinguishing attacks exploiting GPM can achieve near-perfect success rates under suitable parameter choices. \textbf{(Section~\ref{sec:exp})}
\end{itemize}
\section{Preliminaries}\label{sec:prelim}

\subsection{Assumptions and Notations}\label{sec:prelim-assnot}

In this study, we assume that all involved parties, including servers, downstream users, supply-chain backdoor planters, backdoor privacy attackers, regular privacy attackers, and generic adversaries, are probabilistic polynomial-time (PPT). We also assume that all mechanisms discussed in this study are computable in polynomial time. The notations used throughout the paper are summarized in Table~\ref{tab:notation}.

\begin{table}[!t]
    \centering
    \caption{Notations}
    \begin{tabularx}{\linewidth}{lX}
    \toprule
        Notation & Definition \\
        \midrule
        $D\simeq D'$ & $D$ and $D'$ are neighbouring databases\\
        $\norm{\cdot}$ & the $\ell^2$ norm in Euclidean space\\
        $y\sample S$ & $y$ is independently sampled from the uniform distribution over a set $S$\\
        $y\sample \mathcal{P}$ & $y$ is independently sampled from probability distribution $\mathcal{P}$\\
        $y\sample \mathcal{F}(x)$ & $y$ is independently sampled from the output distribution of a randomized function $\mathcal{F}(x)$\\
        $Y\sim S/\mathcal{P}/\mathcal{F}(x)$ & a random variable $Y$ follows the probability distribution as defined above\\
        $\mathcal{P}_1 \equiv \mathcal{P}_2$ & two probability distributions $\mathcal{P}_1$ and $\mathcal{P}_2$ are equivalent\\
        $\negl[\kappa]$ & the class of functions that are asymptotically smaller than any inverse polynomial, i.e., $\left\{f: 0 \leq f(\kappa) < \kappa^{-O(1)}\right\}$\\
        $n$ & sample size\\
        $d$ & dimensionality of mechanism output\\
        \bottomrule
    \end{tabularx}%
    
    \label{tab:notation}%
\end{table}

\subsection{Privacy Notions}\label{sec:prelim-priv}

We first recall \emph{differential privacy (DP)} (Definition~\ref{def:dp}), which is based on the statistical difference of mechanism outputs on neighbouring input databases.

\begin{Def}[differential privacy~\cite{DBLP:journals/fttcs/DworkR14, DBLP:conf/tcc/DworkMNS06}]\label{def:dp}
    Given $\epsilon > 0$ and $0 \leq \delta < 1$, a mechanism $\mathcal{M}:\mathcal{D}\to\mathcal{Y}$ is $\left(\epsilon, \delta\right)$-differentially private (DP) if, for any pair of neighbouring databases $D\simeq D'$ that differ in a single row, and any measurable subset $S\subseteq \mathcal{Y}$, \begin{equation}
        \prob{M\left(D\right) \in S}\leq e^{\epsilon}\prob{\M\left(D'\right) \in S} + \delta.
    \end{equation} 
\end{Def}

\subsection{Gaussian Mechanism}\label{sec:prelim-gm}

The Gaussian mechanism (GM)~\cite{DBLP:journals/fttcs/DworkR14} is one of the most prevalent DP mechanisms. We state its definition and provide its DP guarantee in Definition~\ref{def:gm} and Theorem~\ref{thm:gm-dp}, respectively.

\begin{Def}[Gaussian mechanism] \label{def:gm}
    Given a query $q: \mathcal{D} \to \R^d$, the \emph{Gaussian mechanism (GM)} is defined as the $d$-dimensional multivariate Gaussian distribution with mean $q(D)$ and covariance $\sigma^2I_d$, i.e.,
    \begin{equation}
        \mathcal{M}_{\sigma}^*(D) = \mathcal{N}\left(q(D), \sigma^2I_d\right),
    \end{equation}
    where $\sigma > 0$ and $I_d$ is the $d$-dimensional identity matrix. We denote the probability density function (p.d.f.) of $\mathcal{M}_\sigma^*(D)$ as $f^*_\sigma(D)$.
\end{Def}

\begin{Thm} \label{thm:gm-dp}
    For any $0 < \delta \leq 0.5$ such that 
    $\epsilon = \frac{\Delta^2}{2\sigma^2}- \frac{\Delta}{\sigma}\Phi^{-1}(\delta)$, 
    where $\Delta:= \sup_{D\simeq D'} \norm{q(D) - q(D')}$ is the sensitivity of query $q$, and $\Phi$ is the cumulative distribution function (c.d.f.) of the standard Gaussian distribution, $\mathcal{M}_\sigma^*(\cdot)$ is $(\epsilon, \delta)$-DP~\cite{DBLP:journals/fttcs/DworkR14}.
\end{Thm}

Note that, compared with Theorem~\ref{thm:gm-dp}, a more commonly used yet looser closed-form upper bound for $\epsilon$ is $\frac{\Delta}{\sigma}\sqrt{2\ln \frac{1.25}{\delta}}$ for $0 < {\epsilon, \delta} < 1$. However, due to the extreme values of $\epsilon$ involved in our analysis, we do not restrict ourselves to $\epsilon < 1$ in this study. 

\subsubsection{Numerical Issues and Distinguishing Attacks}\label{sec:prelim-gm-numda}

Despite the theoretical guarantees, actual floating-point implementations of GM have been found to result in enlarged query sensitivities~\cite{DBLP:conf/ccs/Mironov12,DBLP:conf/ccs/CasacubertaSVW22}, and distinguishing attacks have been shown effective in exploiting this numerical issue. In particular, given a pair of neighbouring databases $D\simeq D'$ (with a randomly chosen record replaced) and a query result perturbed by Gaussian noise, the attacker's goal is to determine whether the result was computed on $D$ or $D'$~\cite{DBLP:conf/sp/JinMRO22}. The existence of such attacks has necessitated security patches in DP libraries by further obfuscating the least significant bits~\cite{DBLP:conf/ccs/HolohanB024}.

\subsubsection{Discrete Gaussian Mechanism (DGM)}\label{sec:dgm}

For queries with discrete outputs, one possible mitigation against the numerical issues of GM is the discrete Gaussian mechanism.

\begin{Def}[discrete Gaussian mechanism~\cite{dg}] \label{def:dgm}
    Given a query $q: \mathcal{D} \to \Z$, the \emph{discrete Gaussian mechanism (DGM)} is defined as the discrete Gaussian distribution
    \begin{equation}
        \M^\Z_{\sigma}\left(D\right) := \mathcal{N}_\Z\left(q\left(D\right), \sigma^2\right)
    \end{equation}
    with probability mass function (p.m.f.)
    \begin{equation}
        \phi_\Z\left(z; q\left(D\right), \sigma^2\right)\propto \exp\left(-\frac{\left(z - q\left(D\right)\right)^2}{2\sigma^2}\right)
    \end{equation}
    for any $z\in \Z$, with location $q(D)$ and scale parameter $\sigma$. 
\end{Def}

DGM has various applications, especially in distributed data analytics~\cite{ddgm}. Nevertheless, the use of floating-point GM remains ubiquitous, particularly for queries with non-discrete outputs such as in DP-SGD~\cite{DBLP:conf/ccs/AbadiCGMMT016}.

\subsection{Gaussian Pancake Distribution}

The \emph{Gaussian pancake distribution}, also known as the \emph{continuous learning-with-error (CLWE) distribution}~\cite{clwe, DBLP:conf/focs/GupteVV22}, is the continuous analogue of the \emph{learning-with-error (LWE)} distribution~\cite{DBLP:conf/stoc/Regev05, DBLP:conf/eurocrypt/LyubashevskyPR10}. While the LWE distribution has been widely used for encryption algorithms~\cite{DBLP:journals/jacm/Regev09, DBLP:conf/stoc/Peikert09, DBLP:conf/stoc/PeikertW08, DBLP:conf/stoc/GentryPV08}, CLWE has primarily been applied in learning theory~\cite{DBLP:conf/icml/DiakonikolasKR23, DBLP:conf/colt/DiakonikolasKPZ23, DBLP:conf/nips/DiakonikolasKRS23}, especially in state-of-the-art diffusion models~\cite{DBLP:conf/iclr/ChenC0LSZ23, DBLP:conf/nips/ShahCK23}. Notably, in 2022, it was discovered that indistinguishable backdoors exist for certain deep neural network structures, and one of the constructions relies on the CLWE distribution~\cite{DBLP:conf/focs/GoldwasserKVZ22}. In this study, it suffices to focus on a special case of the CLWE distribution, namely the homogeneous CLWE (hCLWE) distribution, introduced in Definition~\ref{def:hclwe}.

\begin{Def}[hCLWE distribution~\cite{clwe}] \label{def:hclwe}
    For any shape parameter $s > 0$ and $d$-dimensional real vector $\mathbf{x}\in \R^d$, let $\rho_s\left(\mathbf{x}\right):=\exp\left(-\pi\norm{\frac{\mathbf{x}}{s}}^2\right)$, so that $\frac{\rho_s(\mathbf{x})}{s^d}$ is the p.d.f. of the $d$-dimensional multivariate Gaussian distribution with covariance $\frac{s^2}{2\pi}I_d$. The \emph{homogeneous CLWE distribution} $\mathcal{H}_{\mathbf{w}, \beta, \gamma}$ is parameterized by $\mathbf{w}\in \mathbb{S}^{d-1}:= \left\{\mathbf{y}\in \R^d: \norm{\mathbf{y}}=1\right\}$ and $\beta, \gamma > 0$, such that the p.d.f. $\psi_{\mathbf{w}, \beta, \gamma}: \R^d \to \R_+$ is defined for any $\mathbf{y}\in \R^d$ as
    \begin{equation} \label{eq:hclwe-pdf}
        \psi_{\mathbf{w}, \beta, \gamma}(\mathbf{y}) \propto \rho_1(\mathbf{y})\cdot \sum_{z\in \Z}\rho_\beta\left(z - \gamma \mathbf{y}^\top \mathbf{w}\right).
    \end{equation}
\end{Def}

It is also worth noting that $\mathcal{H}_{\mathbf{w}, \beta, \gamma}$ is equivalent to a mixture of Gaussians with equally spaced means and identical variance, as stated in Lemma~\ref{lem:hclwe-mog} and illustrated in Figure~\ref{fig:g-hclwe-comp}. In particular, its Gaussian components are equally spaced along the direction of $\mathbf{w}$, with the distance between neighbouring components given by $\bigTheta{\frac{1}{\gamma}}$. Each Gaussian component has a width of $\bigTheta{\frac{\beta}{\gamma}}$ along $\mathbf{w}$, which is significantly smaller than the $\bigTheta{1}$ width in the orthogonal directions. Consequently, the probability density of hCLWE resembles pancakes stacked along the direction of $\mathbf{w}$, giving rise to the alias \textit{Gaussian pancake distribution}.

\begin{figure}[!t]
    \centering
    \includegraphics[width=0.9\linewidth]{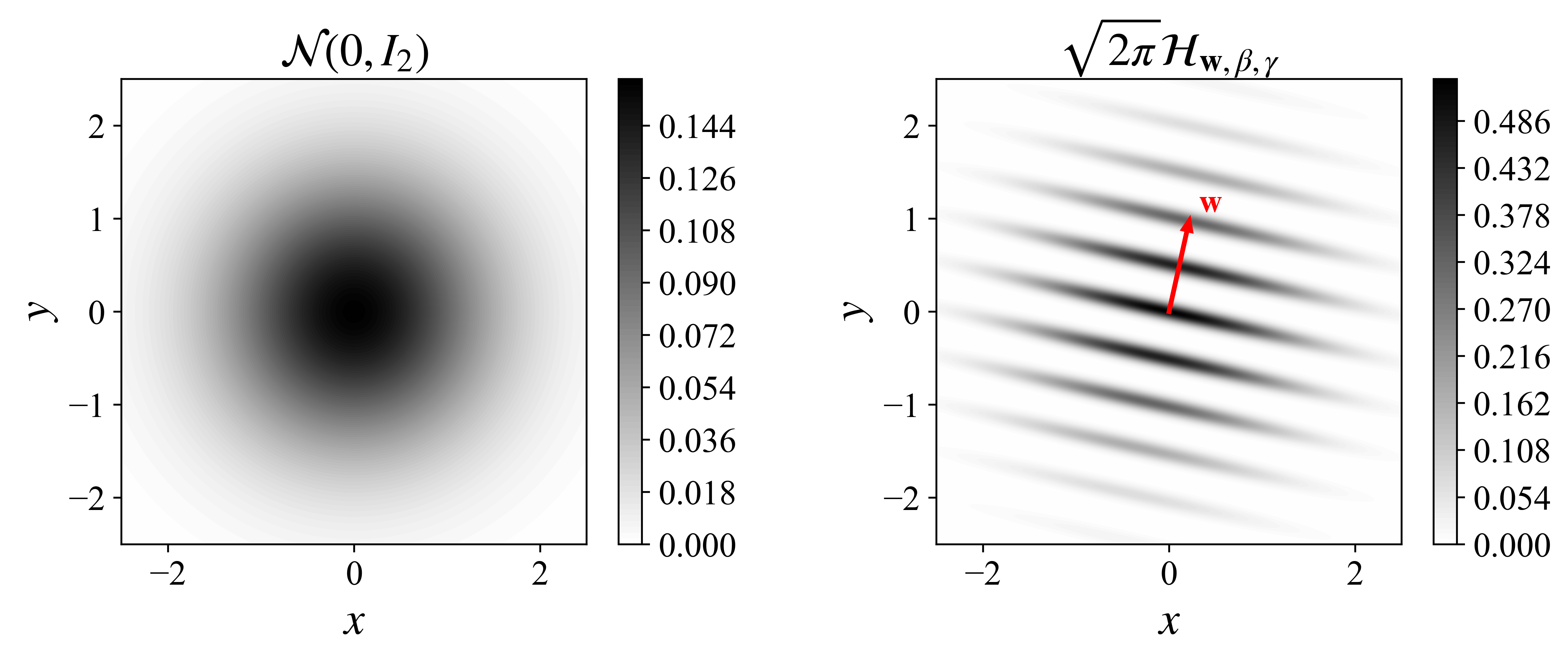}
    \caption{Probability densities of the Gaussian distribution and the hCLWE distribution in $d=2$. Note that distinguishing between the two distributions is assumed to be a hard problem for higher $d$, where $d$ can be viewed as an instantiation of the security parameter $\kappa$.}
    \label{fig:g-hclwe-comp}
\end{figure}

\begin{Lem}\label{lem:hclwe-mog}
    For any $\mathbf{w} \in \mathbb{S}^{d-1}$ and $\beta, \gamma > 0$, the p.d.f. of the hCLWE distribution $\mathcal{H}_{\mathbf{w}, \beta, \gamma}$ (as defined in Definition~\ref{def:hclwe}) can be expressed as
    \begin{equation}
        \psi_{\mathbf{w}, \beta, \gamma}(\mathbf{y}) \propto \sum_{z\in \Z}a_{\beta, \gamma, z}\cdot \phi\left(\mathbf{y}; \frac{\gamma z}{\beta^2 + \gamma^2}\mathbf{w}, \bs{\Sigma}_{\mathbf{w}, \beta, \gamma}\right)\label{eq:hclwe-mog}
    \end{equation}
    where $a_{\beta, \gamma, z} := \exp\left(-\frac{\pi z^2}{\beta^2 + \gamma^2}\right)$ and $\bs{\Sigma}_{\mathbf{w}, \beta, \gamma}:=\frac{1}{2\pi}\left(I - \frac{\gamma^2}{\beta^2 + \gamma^2}\mathbf{w}\mathbf{w}^\top\right)$, such that $\phi\left(\cdot; \frac{\gamma z}{\beta^2 + \gamma^2}\mathbf{w}, \bs{\Sigma}_{\mathbf{w}, \beta, \gamma}\right)$ is the p.d.f. of the multivariate Gaussian distribution $\mathcal{N}\left(\frac{\gamma z}{\beta^2 + \gamma^2}\mathbf{w}, \bs{\Sigma}_{\mathbf{w}, \beta, \gamma}\right)$~\cite{clwe}.
\end{Lem}

The hCLWE distribution and the multivariate Gaussian distribution are considered \textbf{computationally indistinguishable by polynomial-time sampling}, under the hardness assumption that at least one of the \emph{Shortest Independent Vector Problem (SIVP)} or the \emph{Gap Shortest Vector Problem (GapSVP)} is not solvable in bounded-error quantum polynomial time (BQP). Specifically, given dimensionality $d$, an adversary is tasked with distinguishing between the $d$-dimensional hCLWE distribution and the $d$-dimensional multivariate Gaussian distribution, with a sample $s=\left(\mathbf{r}_1, \mathbf{r}_2, \dots, \mathbf{r}_{n\left(d\right)}\right)$ of sample size $n(d)$ polynomially large in $d$, and without knowing the value of $\mathbf{w}$. Theorem~\ref{thm:hclwe-ind-c} states that, under the aforementioned hardness assumptions, no such PPT adversary exists. 

\begin{Thm}\label{thm:hclwe-ind-c}
    Under the assumption that $\mathbf{SIVP}\notin \sf{BQP}$ or $\mathbf{GapSVP}\notin \sf{BQP}$, for $2\sqrt{d}\leq \gamma(d) \leq d^{O(1)}$, $\beta(d)= d^{-O(1)}$, and $n(d) \in d^{O(1)}$, and any family of PPT adversaries $\left\{\adv_d\right\}_{d\in \N}$, there exists $\mu\left(d\right)\in \negl[d]$ such that 
    \begin{equation}
        \abs{\prob{\adv_d\left(s_d\right) = 1:\\s_d\sample {\mathcal{N}\left(\mathbf{0}, I_d\right)}^{\otimes n\left(d\right)}}-\prob{\adv_d\left(s_d\right) = 1:\\\mathbf{w}_d\sample \mathbb{S}^{d-1}\\s_d\sample \sqrt{2\pi}\mathcal{H}_{\mathbf{w}_d, \beta\left(d\right), \gamma\left(d\right)}^{\otimes n(d)}}}\leq \mu\left(d\right),
    \end{equation} 
    where $\otimes n\left(d\right)$ represents the Cartesian product of $n\left(d\right)$ identical distributions. Equivalently, each observation $\mathbf{r}_i\,\left(1\leq i \leq n\left(d\right)\right)$ is independently sampled from the identical distribution of either $\mathcal{N}\left(\mathbf{0}, I_d\right)$ or $\sqrt{2\pi}\mathcal{H}_{\mathbf{w}_d, \beta\left(d\right), \gamma\left(d\right)}$~\cite{clwe}.
\end{Thm}

\section{Problem Setup}\label{sec:setup}

With the hCLWE distribution (Definition~\ref{def:hclwe}), which is computationally indistinguishable from the multivariate Gaussian distribution, we construct the \emph{Gaussian pancake mechanism (GPM)} as an analogy to GM. In particular, GPM perturbs the query result $q\left(D\right)$ by sampling the additive noise from the hCLWE distribution instead of the multivariate Gaussian distribution, as formalized in Definition~\ref{def:gpm}.

\begin{Def}[Gaussian pancake mechanism]\label{def:gpm}
    Given a query $q: \mathcal{D} \to \R^d$ and hyperparameters $\sigma > 0$, $\mathbf{w} \in \mathbb{S}^{d-1}$, and $\beta, \gamma > 0$, the \emph{Gaussian pancake mechanism} is given by
    \begin{equation}
        \label{eq:gpm}
        \mathcal{M}_{\sigma, \mathbf{w}, \beta, \gamma}(D) = q(D) + \sqrt{2\pi}\sigma \cdot \mathcal{H}_{\mathbf{w}, \beta, \gamma}.
    \end{equation}
    We denote the p.d.f. of $\mathcal{M}_{\sigma, \mathbf{w}, \beta, \gamma}(D)$ as $f_{\sigma, \mathbf{w}, \beta, \gamma}(\cdot; D)$.
\end{Def}

Note that the multiplicative factor of $\sqrt{2\pi}$ in the hCLWE noise ensures that the scale remains consistent with GM. In the rest of this section, we formalize a backdoor attack based on GPM, while formalizing its undetectability by any party other than the backdoor key.

\subsection{Threat Model}\label{sec:setup-threat-model}

\begin{figure}
    \centering
    \includegraphics[width=0.8\linewidth]{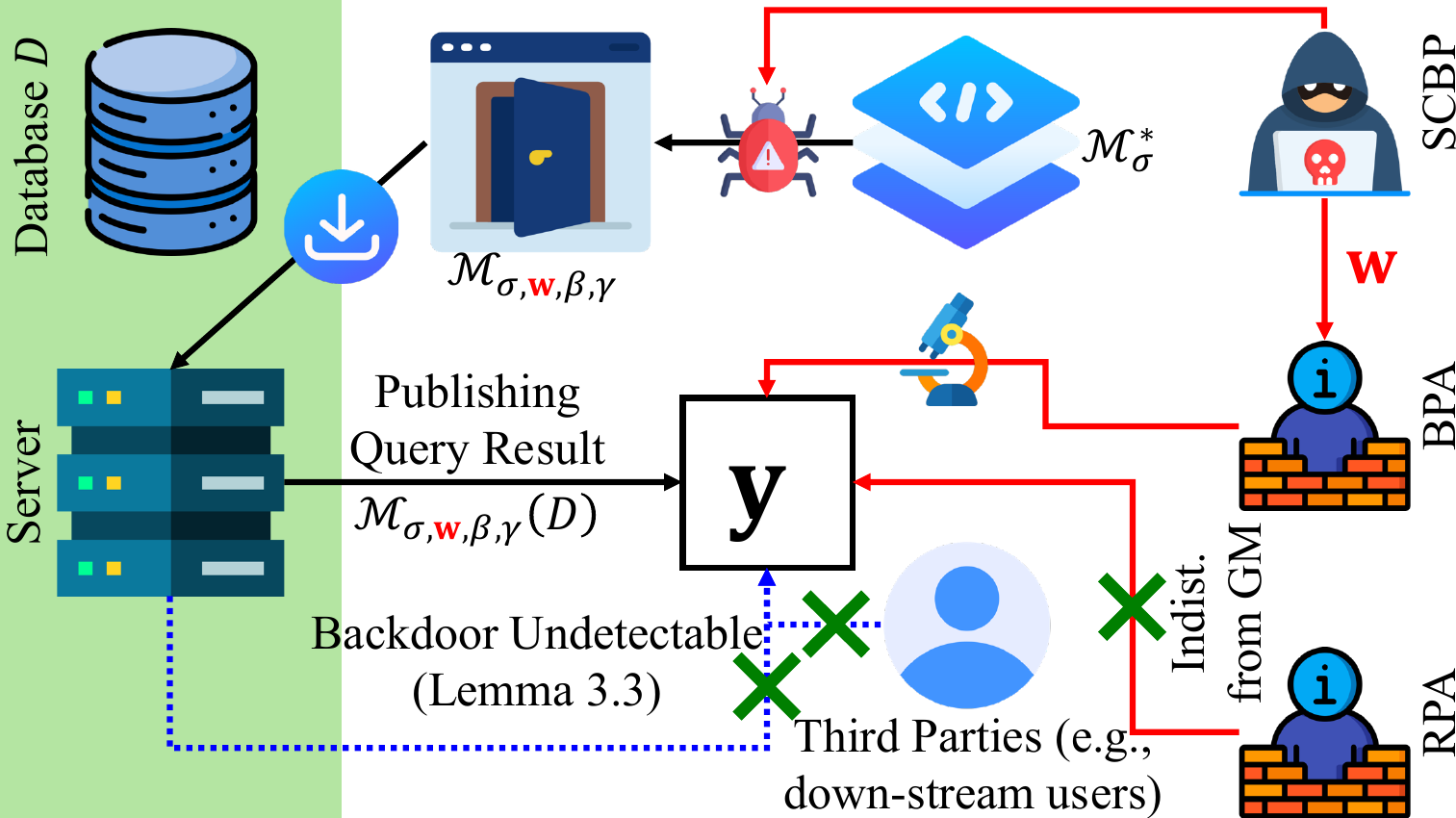}
    \caption{Illustration of the Backdoor Attack Model: the supply-chain backdoor planter (SCBP) plants the backdoor by replacing GM with GPM. Consequently, the backdoor privacy attacker (BPA) gains information on $D$ using GPM's output $\mathbf{y}$ and knowledge of the secret backdoor key $\mathbf{w}$. However, against a regular privacy attacker (RPA) without knowledge of $\mathbf{w}$, the GPM output is computationally indistinguishable from GM, thus preserving the same privacy guarantees. For the same reason, neither the server nor other parties (e.g., downstream users) can detect the existence of the backdoor.}
    \label{fig:threat-model}
\end{figure}

We consider the threat model presented in Figure~\ref{fig:threat-model}, where a \textbf{supply-chain backdoor planter (SCBP)} installs the backdoor in the compromised software, alongside two types of PPT privacy attackers: the \textbf{backdoor privacy attacker (BPA)} and the \textbf{regular privacy attacker (RPA)}. Additionally, we assume that the \textbf{server and third parties (e.g., downstream users) monitor anomalies in the server's outputs}. Specifically, the SCBP, the BPA and the RPA follow the attack vector below:

\begin{itemize}[leftmargin=*]
    \item Given online software that utilizes differential privacy (e.g., Apple macOS, Amazon AWS, and Microsoft Azure), the SCBP targets the implementation of the authentic GM $\M_\sigma^*:\mathcal{D} \to \R^d$. It computes $\beta=\beta\left(d\right)$ and $\gamma=\gamma\left(d\right)$ that satisfy the conditions in Theorem~\ref{thm:hclwe-ind-c}, samples $\mathbf{w}\sample \mathbb{S}^{d-1}$ (acting as the secret \textbf{backdoor key}), plants the backdoor by replacing the authentic $\M_\sigma^*$ with GPM $\M_{\sigma, \mathbf{w}, \beta, \gamma}$. Then, it shares the backdoor key $\mathbf{w}$, as well as the other configuration parameters $\beta$ and $\gamma$, with the BPA.
    \item The server, in possession of a sensitive database $D$, downloads the compromised software. However, the compromised software is not open-source or the server fails to fully inspect the source code. It executes GPM as $\mathbf{y}\sample \mathcal{M}_{\sigma, \mathbf{w}, \beta, \gamma}\left(D\right)$, and releases $\mathbf{y}$. Given the privacy requirements, we do \textbf{not} assume that the compromised software has access to the sensitive database $D$ and the unperturbed query results $q\left(D\right)$.
    \begin{Exm} \label{exm:randomness-lib-compromised}
        If the server stores $D$ and computes $q\left(D\right)$ in a protected environment while sampling the additive Gaussian noise from a randomness library, then compromising the randomness library by replacing the multivariate Gaussian distribution with hCLWE is sufficient to implant the backdoor, as the output will be equivalent to GPM. However, in this scenario, the randomness library need not (and should not) have access to $D$ and $q\left(D\right)$.
    \end{Exm}
    \item The BPA, with knowledge of the backdoor key $\mathbf{w}$, computes additional information about $D$ from $\mathbf{y}$. Meanwhile, the RPA also attempts to compute additional information about $D$ but without the knowledge of the backdoor key $\mathbf{w}$. 
\end{itemize}

Note that the \emph{SCBP could be the software developer itself} or other malicious attackers. Specifically, in addition to the prevalent numerical issues in DP libraries~\cite{DBLP:conf/sp/JinMRO22, DBLP:conf/ccs/Ilvento20, DBLP:conf/ccs/Mironov12, DBLP:conf/ccs/CasacubertaSVW22, DBLP:conf/uss/HaeberlenPN11}, the following real-world incidents showcase the prevalence and severity of the risk that DP and other privacy-sensitive software could be compromised, either via official channels or through software piracy and DNS hijacking:\begin{itemize}[leftmargin=*]
    \item \textbf{Early versions of MacOS}, despite the deployment of DP mechanisms for user privacy protection, were found to have $\epsilon\geq 16$, which offers only limited protection~\cite{DBLP:journals/corr/abs-1709-02753}.
    \item The Lazarus Group committed multiple cryptocurrency thefts by utilizing \textbf{compromised software} loaded with \textbf{incorrect keys and malicious smart contracts}~\cite{DBLP:journals/corr/abs-2505-21725}.
    \item A \textbf{backdoor} was suspected to have been inserted into \textbf{a published pseudorandom number generator standard}, allowing adversaries to readily decrypt materials encrypted using it~\cite{DBLP:conf/crypto/BrownG07}.
\end{itemize}

Additionally, unlike the scenarios in Example~\ref{exm:counterexample-threat-model}, the compromised software has the same behaviour as the authentic one, except for the different yet indistinguishable noise distribution. Therefore, the server does not observe abnormalities in the execution trace that could lead to the detection of the backdoor.

\begin{Exm}\label{exm:counterexample-threat-model}
    Consider the following alternative methods of compromising the software and breaking the privacy guarantees:
    \begin{itemize}[leftmargin=*]
        \item If the compromised software employs steganography~\cite{DBLP:conf/crypto/HopperLA02} to embed additional information about $D$ or $q\left(D\right)$ into the output $\mathbf{y}$, then additional reading authorization for these protected entities shall be requested. This causes an abnormal access pattern in the scenario of Example~\ref{exm:randomness-lib-compromised}.
        \item If the BPA exploits the random seed to recover the additive noise, then either \textbf{1)} the random seed is embedded in the software, which produces consistent noise and is therefore detectable, or \textbf{2)} the BPA needs to monitor the internal state of the server (e.g., timestamps and/or all calls to the pseudorandom generator) in real time to compute the random seed, which is unrealistic in general and defendable by a trivial reset on the server's side.
    \end{itemize}
\end{Exm}

\subsection{Undetectability of Backdoor}\label{sec:setup-covertness}

In this section, we characterize the covertness of the GPM under the threat model introduced in Section~\ref{sec:setup-threat-model}. In particular, Theorem~\ref{thm:gpm-covert} captures the case where the server downloads compromised software containing GPM $\mathcal{M}_{\sigma, \mathbf{w}, \beta, \gamma}$ after the SCBP samples $\mathbf{w}$. The server then executes the mechanism \textbf{polynomially many} (rather than a single) times with input databases $D_1, D_2, \dots, D_n$ (where $n=n\left(d\right)$ is a polynomial in $d$), respectively, and releases the results $\mathbf{y}_1, \mathbf{y}_2, \dots, \mathbf{y}_n$. We formalize that, by only examining these results, no party other than the SCBP and the BPA (i.e., one not knowing $\mathbf{w}$, including the server itself, downstream users, or any third party) can distinguish them from the authentic output of GM on the same sequence of input databases.

\begin{Thm}[Covertness of GPM backdoor]\label{thm:gpm-covert}
    For any family of PPT adversaries $\left\{\adv_d\right\}_{d\in \N}$, any $n\left(d\right)\in d^{O(1)}$, any $2\sqrt{d}\leq \gamma(d) \leq d^{O(1)}$, any $\beta(d)= d^{-O(1)}$, there exists $\mu\left(d\right)\in \negl[d]$ such that for any $D_1, D_2, \dots, D_{n(d)} \in \D_d$, \begin{equation}
        \abs{\prob{\adv_d\left(o_d\right) = 1:\\ o_d \sample \sf{GMSeq}(d)} - \prob{\adv_d\left(o_d\right) = 1:\\ o_d \sample \sf{GPMSeq}(d)}} \leq \mu\left(d\right),\label{eq:gpm-covert}
    \end{equation} where $\sf{GMSeq}(d)$ and $\sf{GPMSeq}(d)$ are defined as Figure~\ref{fig:gmseq-gpmseq}.

    \begin{figure}
        \centering
        \begin{pchstack}[center]
            \procedure[space=auto,linenumbering]{$\sf{GMSeq}(d)$}{\pcfor 1\leq i \leq n\left(d\right)\\
                \mathbf{y}_i \sample \mathcal{M}^*_{\sigma(d)}\left(D_i\right)\\
            \pcendfor\\
            \pcreturn \mathbf{y}_1, \mathbf{y}_2, \dots, \mathbf{y}_{n\left(d\right)}}
    
            \procedure[space=auto,linenumbering]{$\sf{GPMSeq}(d)$}{\label{pcln:gpmseq-samplekey}\mathbf{w}_d\sample\mathbb{S}^{d-1}\\
            \label{pcln:gpmseq-execbegin}\pcfor 1\leq i \leq n\left(d\right)\\
                \mathbf{y}_i \sample \mathcal{M}_{\sigma(d),\mathbf{w}_d,\beta(d), \gamma(d)}\left(D_i\right)\\
            \label{pcln:gpmseq-execend}\pcendfor\\
            \pcreturn \mathbf{y}_1, \mathbf{y}_2, \dots, \mathbf{y}_{n\left(d\right)}
            }
        \end{pchstack}
        \caption{Definitions of $\sf{GMSeq}(d)$ and $\sf{GPMSeq}(d)$.}
        \label{fig:gmseq-gpmseq}
    \end{figure}
\end{Thm}

In Theorem~\ref{thm:gpm-covert}, $\adv_d$ represents any party (e.g., the server and downstream users) trying to distinguish between the outputs of GM and GPM, without knowledge of the backdoor key $\mathbf{w}$. Meanwhile, $\sf{GMSeq}(d)$ corresponds to the execution of the authentic GM, where the server applies GM multiple times on a sequence of input databases $D_1, D_2, \dots, D_{n(d)}$, yielding outputs $\mathbf{y}_1, \mathbf{y}_2, \dots, \mathbf{y}_{n(d)}$, respectively. Meanwhile, $\sf{GPMSeq}(d)$ corresponds to the planting and execution of GPM, where in Line~\ref{pcln:gpmseq-samplekey}, the SCBP samples its secret backdoor key $\mathbf{w}$ upon compromising the DP library. Subsequently, in Lines~\ref{pcln:gpmseq-execbegin} to \ref{pcln:gpmseq-execend}, the server executes GPM to output $\mathbf{y}_1, \mathbf{y}_2, \dots, \mathbf{y}_{n(d)}$ on the same sequence of input databases. As long as GPM is executed polynomially many times (i.e., $n(d)$ is polynomially bounded), no party can efficiently distinguish its outputs from those of the authentic GM.

\begin{proof}[Proof sketch of Theorem~\ref{thm:gpm-covert}]
    Assume the opposite. Then an efficient distinguishing adversary could distinguish between samples from the multivariate Gaussian distribution and the hCLWE distribution, thereby breaking the hardness assumption.
\end{proof}

Similar to Theorem~\ref{thm:gpm-covert}, for the RPA without knowledge of the backdoor key $\mathbf{w}$, GPM maintains the same level of privacy guarantee, which is formalized in \appref{app:setup-regular-attacker-cdp}.
\section{Analysis on the Statistical Privacy Leakage to the BPA}\label{sec:privacy}

By Theorem~\ref{thm:gpm-covert}, GPM is computationally indistinguishable from GM and achieves the same level of privacy computationally, i.e., against a PPT adversary without knowledge of the backdoor key $\mathbf{w}$. However, as shown in Figure~\ref{fig:gpm-density}, in the direction of $\mathbf{w}$, the probability density of GPM on neighbouring databases could be well separated, which creates an opportunity for the BPA who knows $\mathbf{w}$. Therefore, in this section, we proceed to show the actual statistical privacy loss, which is arbitrarily larger than that of GM. In particular, as shown in Figure~\ref{fig:gpm-density}, for neighbouring databases $D\simeq D'$, due to the differences between the centres $q(D)$ and $q\left(D'\right)$, the peaks of the GPM probability densities may be misaligned, which therefore voids the DP guarantee provided by GM, despite the computational indistinguishability. Nevertheless, as GPM is still composed of multiple Gaussian components, it can be expected that, although significantly degraded compared with GM, some statistical DP guarantees can still be offered.

\begin{figure}[!t]
    \centering
    \includegraphics[width=0.6\linewidth]{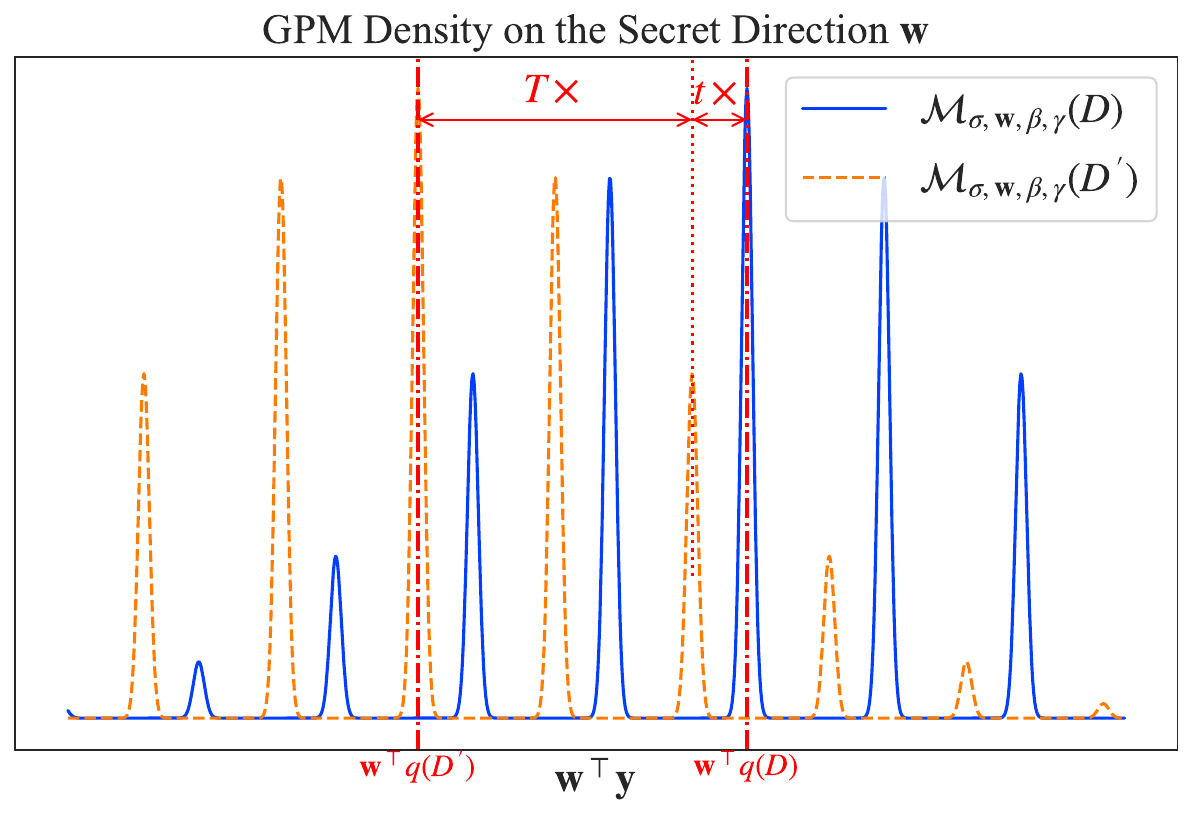}
    \caption{When projected onto the secret direction $\mathbf{w}$, the probability densities of $\M_{\sigma, \mathbf{w}, \beta, \gamma}\left(D\right)$ and $\M_{\sigma, \mathbf{w}, \beta, \gamma}\left(D'\right)$ for neighbouring databases $D\simeq D'$ are concentrated in disjoint sets of intervals and are therefore well separated, breaking the original DP guarantee. In this example, $q\left(D\right)$ and $q\left(D'\right)$ differ by $2.4$ peak widths, such that $T = 2$ and $t = 0.4$.}
    \label{fig:gpm-density}
\end{figure}

\begin{figure*}[!t]
    \centering
    \includegraphics[width=\linewidth]{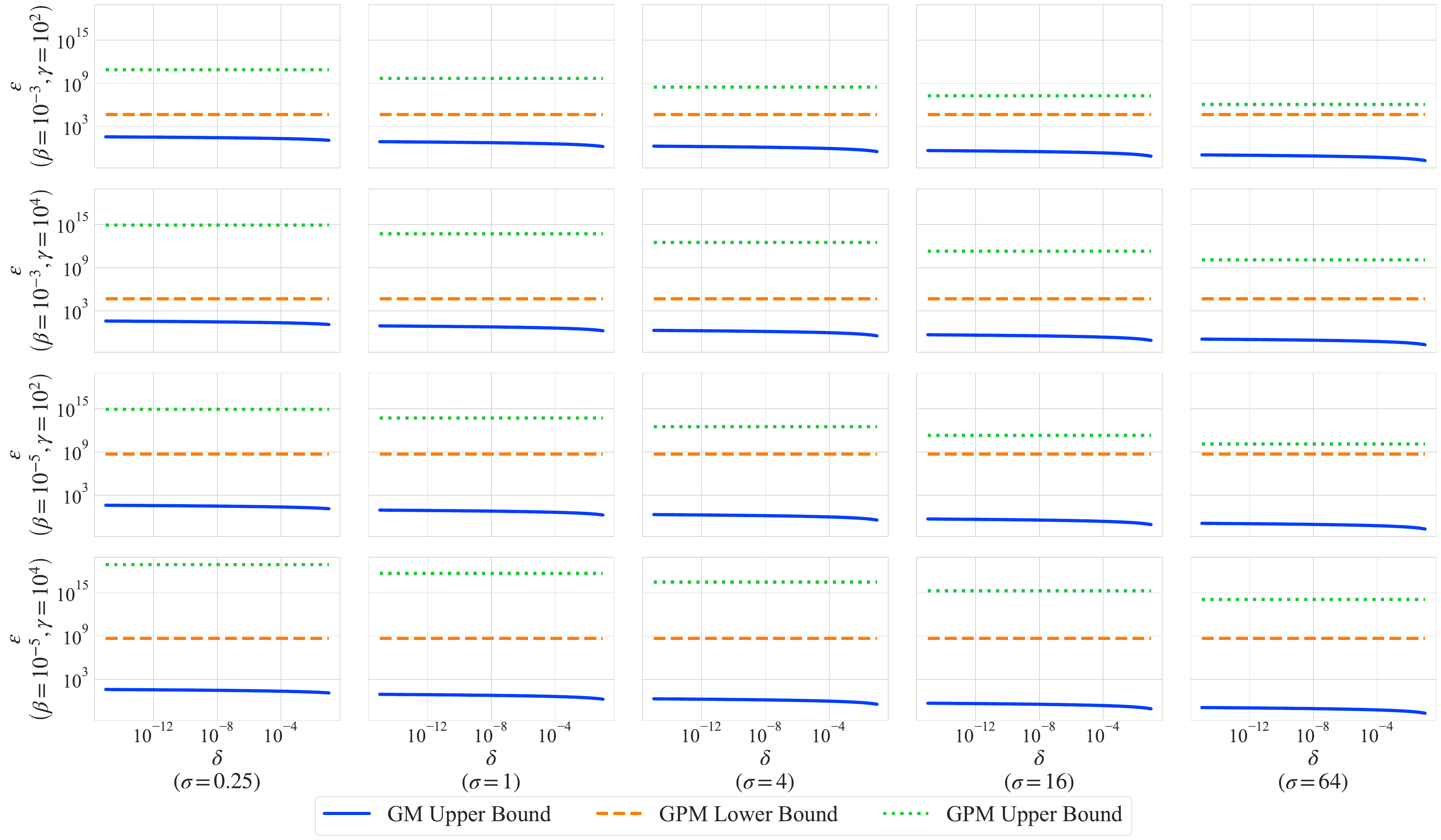}
    \caption{Comparison between GM and GPM's privacy costs at $t = 0.25$, as in Theorem~\ref{thm:gpm-dp-lb}. Although computationally indistinguishable from GM, GPM's actual privacy costs exceed $\epsilon=10^4$ even when $\delta = 0.1$, which is infeasible for practical applications. Note that the lower and upper bounds of GPM's privacy costs still change with respect to $\delta$, although at smaller relative scales.}
    \label{fig:privacy-bounds}
\end{figure*}

\subsection{Lower Bound}\label{sec:privacy-lb}

To exactly analyze the additional privacy leakage arising from the misalignments of the peaks of probability densities, we first note that, by Lemma~\ref{lem:hclwe-mog}, the marginal distributions of the Gaussian components have the narrowest peaks when projected onto the secret direction $\mathbf{w}$. Therefore, as depicted in Figure~\ref{fig:gpm-density}, when executing the queries on different databases, the peaks of probability density functions could be significantly misaligned. This misalignment of the peaks causes additional privacy leakage compared to GM, where the privacy cost originates from the shift of the means in the Gaussian distributions~\cite{DBLP:journals/corr/abs-1905-02383,DBLP:journals/fttcs/DworkR14}, and can thus be exploited by the BPA with knowledge of $\mathbf{w}$. To simplify the analysis, we first project $\mathcal{H}_{\mathbf{w}, \beta, \gamma}$ onto the direction of $\mathbf{w}$, which is equivalent to $\mathcal{H}_{1, \beta, \gamma}$, with p.d.f.
\begin{equation}
    \psi_{1, \beta, \gamma}\left(y\right) \propto \sum_{z\in \Z}a_{\beta, \gamma, z}\phi\left(y; \frac{\gamma z}{\beta^2 + \gamma^2}, \frac{1}{2\pi}\cdot\frac{\beta^2}{\beta^2 + \gamma^2}\right)\label{eq:hclwe-1-mog},
\end{equation}
and state the key technical Lemma~\ref{lem:one-dimensional-hclwe-diff} for analyzing the behaviour of GPM in the direction of $\mathbf{w}$. As exemplified in Figure~\ref{fig:gpm-density}, $T$ and $t$ correspond to the integer and decimal parts of the ratio between the projected difference of the query result $\mathbf{w}^\top\left(q\left(D\right)-q\left(D'\right)\right)$ and CLWE peak separation $\sqrt{2\pi}\sigma\cdot\frac{\gamma}{\beta^2+\gamma^2}$. Note that the decomposition into $T$ and $t$ is unique. The proof of Lemma~\ref{lem:one-dimensional-hclwe-diff} is provided in \appref{app:privacy-one-dimensional-hclwe-diff}.

\begin{Lem}\label{lem:one-dimensional-hclwe-diff}
    For any $T \in \Z$ and $t \in [-0.5, 0.5)$, let $Y \sim \mathcal{H}_{1, \beta, \gamma}$ and $Y'\sim \mathcal{H}_{1, \beta, \gamma} + \left(T + t\right)\frac{\gamma}{\beta^2 + \gamma^2}$ be two random variables. There exists a set $A(t):= \bigcup_{z\in \Z}I_z(t)$, where $I_z(t):=\left(\frac{\gamma \left(z-\frac{\abs{t}}{2}\right)}{\beta^2 + \gamma^2}, \frac{\gamma \left(z+\frac{\abs{t}}{2}\right)}{\beta^2 + \gamma^2}\right)$, such that
    \begin{align}
        \prob{Y \in A(t)} \geq& 1 - 2\Phi\left(-\frac{\gamma\abs{t}}{\beta}\sqrt{\frac{\pi}{2\left(\beta^2 + \gamma^2\right)}}\right),\\
        \prob{Y'\in A(t)} \leq& 2\Phi\left(-\frac{\gamma\abs{t}}{\beta}\sqrt{\frac{\pi}{2\left(\beta^2 + \gamma^2\right)}}\right),
    \end{align}
    where $\Phi$ is the c.d.f. of the standard Gaussian distribution.
\end{Lem}

Lemma~\ref{lem:one-dimensional-hclwe-diff} can be applied to the marginal distributions of GPM's outputs on $D$ and $D'$ when projected onto $\mathbf{w}$, which gives rise to the lower bound of the actual statistical privacy leakage of GPM as stated in Theorem~\ref{thm:gpm-dp-lb}. The full proof of Theorem~\ref{thm:gpm-dp-lb} is provided in \appref{app:privacy-lb}.

\begin{Thm}\label{thm:gpm-dp-lb}
    For any $\sigma > 0$, $\mathbf{w}\in \mathbb{S}^{d-1}$, $0 < \beta < 1$, $\gamma > 1$, and any pair $T\in \Z$ and $t\in \left[-0.5, 0.5\right)$, if there exists a pair of neighbouring databases $D\simeq D'$ such that $\sqrt{2\pi} \sigma \cdot (T + t) \frac{\gamma}{\beta^2 + \gamma^2} = \mathbf{w}^\top \left(q\left(D'\right) - q(D)\right)$, then $\mathcal{M}_{\sigma, \mathbf{w}, \beta, \gamma}$ is \textbf{not} $\left(\epsilon, \delta\right)$-DP for any
    \begin{equation}\label{eq:gpm-dp-lb}
        0 < \epsilon < \underline{\epsilon} := \log\left( \frac{1 - \delta}{2\Phi\left(-\frac{\gamma \abs{t}}{\beta} \sqrt{\frac{\pi}{2(\beta^2 + \gamma^2)}}\right)} - 1 \right)
    \end{equation}
    and $0 < \delta \leq 0.5$.
\end{Thm}

\begin{proof}[Proof sketch of Theorem~\ref{thm:gpm-dp-lb}]
    For neighbouring databases $D \simeq D'$, when $q(D)$ and $q\left(D'\right)$ differ by $T + t$ in the direction of $\mathbf{w}$, the projected distributions of $\mathcal{M}_{\sigma, \mathbf{w}, \beta, \gamma}\left(D\right)$ and $\mathcal{M}_{\sigma, \mathbf{w}, \beta, \gamma}\left(D'\right)$ are $\mathcal{H}_{1, \beta, \gamma}$ and $\mathcal{H}_{1, \beta, \gamma} + \left(T + t\right)\frac{\gamma}{\beta^2 + \gamma^2}$, respectively. Therefore, by Lemma~\ref{lem:hclwe-mog}, to achieve $\left(\epsilon, \delta\right)$-DP, it must hold that $\prob{Y \in A(t)} \leq e^{\epsilon}\prob{Y' \in A(t)} + \delta$, where Inequality \eqref{eq:gpm-dp-lb} does not hold. 
\end{proof}

Furthermore, note that the indistinguishability between GM and GPM hinges on the uniform randomness of the parameter $\mathbf{w} \sample \mathbb{S}^{d-1}$. When the separation $\sqrt{2\pi}\sigma\cdot \frac{\gamma}{\beta^2 + \gamma^2}$ between the Gaussian components is sufficiently small, the range of $T + t$ is significantly larger than $1$. Therefore, the distribution of the decimal part $t$ is close to the uniform distribution, such that in the average case, $\abs{t}$ becomes a constant in Inequality \eqref{eq:gpm-dp-lb}. We formalize this in Proposition~\ref{prop:gpm-dp-lb} and provide the detailed analysis in \appref{app:privacy-lb}. 

\begin{Prop}\label{prop:gpm-dp-lb}
    For any $\sigma > 0$, $0 < \beta < 1$, and $\gamma > 1$ such that $\gamma \gg \sqrt{d} \cdot \frac{\sigma}{\Delta}$, with constant probability over sampling $\mathbf{w} \sample \mathbb{S}^{d-1}$, $\mathcal{M}_{\sigma, \mathbf{w}, \beta, \gamma}$ is not $(\epsilon, \delta)$-DP where $\epsilon \in \bigTheta{\frac{1}{\beta^2}}$ for any $0 < \delta \leq 0.5$.
\end{Prop}

However, it is also noteworthy to consider the marginal case where $q\left(D\right)$ and $q\left(D'\right)$ are very close in the direction of $\mathbf{w}$, such that their separation is significantly smaller than one peak. This results in \textbf{$T = 0$ and $t \ll 1$}, which decreases the lower bound of privacy cost by Theorem~\ref{thm:gpm-dp-lb}, and is therefore a situation deemed \textbf{``unfavourable'' from the perspective of the BPA}.

\subsection{Upper Bound}\label{sec:privacy-ub}

Despite the broken original privacy guarantee against the BPA with knowledge of $\mathbf{w}$ shown in Section~\ref{sec:privacy-lb}, GPM still achieves (statistical) DP theoretically, albeit at a significantly higher privacy cost. The intuition behind developing this upper bound is that, due to its nature as a mixture of Gaussian components with the same variance, GPM can still be viewed as first adding Gaussian noise and then post-processing by choosing the mean (i.e., deciding which component it falls into). However, the Gaussian noise is significantly narrower in the direction of $\mathbf{w}$ than in the authentic GM, thereby incurring a significantly larger privacy cost.

\begin{Thm} \label{thm:gpm-dp-ub}
    $\mathcal{M}_{\sigma, \mathbf{w}, \beta, \gamma}(\cdot)$ is $(\epsilon, \delta)$-DP for any $0 < \delta \leq 0.5$ such that 
    $\epsilon \in \bigTheta{\frac{\gamma^2\Delta^2}{\beta^2\sigma^2} + \frac{\gamma\Delta}{\beta\sigma}\sqrt{\log\frac{1}{\delta}}}$.
\end{Thm}

\begin{proof}[Proof Sketch of Theorem~\ref{thm:gpm-dp-ub}]
    Note that $\M_{\sigma, \mathbf{w}, \beta, \gamma}\left(D\right)$ is equivalent to
    \begin{equation}
        q(D) + \sqrt{2\pi}\sigma\left(\mathcal{N}\left(\mathbf{0}, \mathbf{\Sigma}_{\mathbf{w}, \beta, \gamma}\right) +  \frac{\gamma \mathbf{w}\cdot\mathcal{N}_\Z\left(0, \frac{\beta^2+\gamma^2}{2\pi} \right)}{\beta^2+\gamma^2} \right),
    \end{equation}
    where the discrete Gaussian noise can be viewed as post-processing. On the other hand, $\sqrt{2\pi}\sigma\cdot\mathcal{N}\left(\mathbf{0}, \mathbf{\Sigma}_{\mathbf{w}, \beta, \gamma}\right)$ has variance at least $\frac{\beta^2}{\beta^2 + \gamma^2}\sigma^2$ in any direction (the lower bound is achieved in the direction of $\mathbf{w}$). Therefore, by Theorem~\ref{thm:gm-dp}, $\M_{\sigma, \mathbf{w}, \beta, \gamma}\left(D\right)$ is $\left(\epsilon, \delta\right)$-DP for any 
    \begin{equation}
        \epsilon \geq \overline{\epsilon} := \frac{\left(\beta^2 + \gamma^2\right)\Delta^2}{2\beta^2 \sigma^2} - \frac{\sqrt{\beta^2 + \gamma^2}\Delta}{\beta \sigma}\Phi^{-1}\left(\delta\right)
    \end{equation}
    and $0 < \delta \leq 0.5$.
\end{proof}

We present the full proof of Theorem~\ref{thm:gpm-dp-ub} in \appref{app:privacy-ub}. As a sanity check, note that since Proposition~\ref{prop:gpm-dp-lb} requires $\gamma \gg \sqrt{d}\cdot \frac{\sigma}{\Delta} \geq \frac{\sigma}{\Delta}$, it follows that $\frac{\gamma^2\Delta^2}{\beta^2\sigma^2}\gg \frac{1}{\beta^2}$. That is, the upper bound is higher than the lower bound. Moreover, we plot in Figure~\ref{fig:privacy-bounds} the privacy costs under several configurations of hyperparameters at sensitivity $\Delta = 1$. In particular, it can be observed that the value of $\beta$ is the determinant factor of the lower bound of GPM's privacy cost, while $\gamma$ has little impact on it. On the other hand, the upper bound depends on both $\gamma$ and $\beta$. These observations align with Proposition~\ref{prop:gpm-dp-lb} and Theorem~\ref{thm:gpm-dp-ub}. For all choices of hyperparameters listed, the actual value of $\epsilon$ is at least $10^4$, which indicates that no meaningful privacy protection can be achieved with GPM against an adversary that has knowledge of $\mathbf{w}$.
\section{Distinguishing Attacks}\label{sec:da}

Section~\ref{sec:privacy} shows that the actual privacy cost of GPM is asymptotically larger than that of the genuine GM, despite their computational indistinguishability. In this section, we further demonstrate that this privacy leakage translates to additional power for the BPA with knowledge of the backdoor key $\mathbf{w}$. In particular, we construct a concrete attack to distinguish between neighbouring databases $D_0$ and $D_1$, from which the ``privatized'' query results are computed. The intuition behind this attack is that, as depicted in Figure~\ref{fig:gpm-density}, since the GPM densities are concentrated around several peaks in the direction of $\mathbf{w}$, the correct input database can be determined by checking whether the output is close to one of the peaks with very high probability.

Following the same setting for distinguishing attacks introduced in Section~\ref{sec:prelim-gm-numda}, we consider the probability that the BPA $\adv_{\sigma, \mathbf{w}, \beta, \gamma}$ (equipped with knowledge of $\mathbf{w}$) successfully identifies the index of the input database $i\in \bin$, i.e.,
\begin{equation}
    \prob{i^*=i:\\i\sample\bin\\ \mathbf{y}\sample \M_{\sigma, \mathbf{w}, \beta, \gamma}(D_i)\\i^*\sample\adv_{\sigma, \mathbf{w}, \beta, \gamma}\left(D_0, D_1, \mathbf{y}\right)},\label{eq:prob-da-success}
\end{equation}
where $i\in \bin$ is the real index of the database, and $i^*$ is the guess of $i$ by the attacker after observing the output of the mechanism.

By Lemma~\ref{lem:hclwe-mog}, the probability mass of $\M_{\sigma, \mathbf{w}, \beta, \gamma}\left(D_i\right)$ is concentrated around the hyperplanes
\begin{equation}\label{eq:hyperplanes}
    H_i:=\bigcup_{z\in \Z}\left\{\mathbf{y}\in \R^d: \frac{\left(\mathbf{y} - q\left(D_i\right)\right)^\top \mathbf{w}}{\sqrt{2\pi} \sigma} = \frac{\gamma z}{\beta^2 + \gamma^2}\right\},
\end{equation}
for $i \in \bin$, such that $\adv_{\sigma, \mathbf{w}, \beta, \gamma}\left(D_0, D_1, \mathbf{y}\right)$ (Figure~\ref{fig:da-attacker}) can be constructed by determining whether the closest hyperplane to $\mathbf{y}$ belongs to $H_0$ or $H_1$.

\begin{figure}
    \centering
    \procedureblock[space=auto, linenumbering]{$\adv_{\sigma, \mathbf{w}, \beta, \gamma}\left(D_0, D_1, \mathbf{y}\right)$}{
        \pcfor i\gets\bin\\
            z_i \gets \frac{\left(\beta^2 + \gamma^2\right)\cdot\left(\mathbf{y} - q\left(D_i\right)\right)^\top \mathbf{w}}{\sqrt{2\pi}\sigma \cdot \gamma}\\
        \pcendfor\\
        \pcreturn i^*\gets \argmin_{i\in \bin}\abs{z_i - \round{z_i}}
    }
    \caption{Definition of $\adv_{\sigma, \mathbf{w}, \beta, \gamma}$.}
    \label{fig:da-attacker}
\end{figure}

\begin{Thm}\label{thm:da}
    For any pair of neighbouring databases $D_0\simeq D_1$, and the unique pair of $T\in \Z$ and $t\in \left[-0.5, 0.5\right)$ such that $\sqrt{2\pi} \sigma \cdot (T + t) \frac{\gamma}{\beta^2 + \gamma^2} = \mathbf{w}^\top \left(q(D_1) - q(D_0)\right)$, the probability that $\adv_{\sigma, \mathbf{w}, \beta, \gamma}$ correctly identifies the input database (i.e., Equation \eqref{eq:prob-da-success}) is at least $1 - 2\Phi\left(- \frac{\gamma\abs{t}}{\beta}\sqrt{\frac{\pi}{2\left(\beta^2 + \gamma^2\right)}}\right)$.
\end{Thm}

\begin{proof}[Proof of Theorem~\ref{thm:da}]
    Without loss of generality, condition on $i = 0$. Observe that, when $\mathbf{y}$ is closer to $H_0$ than $H_1$, $\mathcal{A}_{\sigma, \mathbf{w}, \beta, \gamma}$ outputs $0$. Following the same argument as in Theorem~\ref{thm:gpm-dp-lb}, this event occurs with probability at least $1 - 2\Phi\left(-\frac{\gamma\abs{t}}{\beta}\sqrt{\frac{\pi}{2\left(\beta^2 + \gamma^2\right)}}\right)$.
\end{proof}

By Theorem~\ref{thm:da}, the success rate of the distinguishing attack is very close to $1$ given $\gamma \gg 1 \gg \beta$, which we further verify empirically in Section~\ref{sec:exp-da}. Here, the failure probability corresponds to the case where the output shifts away from the peaks of the GPM density, and instead falls closer to the peaks corresponding to the other database, which happens with infinitesimal probability under proper choices of parameters.
\section{Mitigations}\label{sec:mit}

In this section, we discuss possible mitigations against the privacy attack enabled by GPM as one major possible direction for future work. In particular, in Section~\ref{sec:mit-rotation}, we discuss a possible mitigation method under the same threat model as considered in Section~\ref{sec:setup-threat-model}, which is fully compatible with the scenario discussed in Example~\ref{exm:randomness-lib-compromised}, where the randomness library is compromised. Meanwhile, in Section~\ref{sec:mit-vdp}, we discuss viable methods of enforcing the correct implementations and executions of the DP mechanisms. We present in \appref{app:mit-distributed} an additional mitigation method under the setting of distributed DP with multiple servers, where a subset of the servers executes the compromised software with GPM.

\subsection{Random Rotation}\label{sec:mit-rotation}

Given that the privacy attack hinges on the BPA's knowledge of $\mathbf{w}$, and that the output is computationally indistinguishable from GM otherwise, one mitigation against GPM is to apply a random rotation to the sampled noise. Specifically, as described in Algorithm~\ref{alg:mit-rotation}, upon sampling the noise value $\mathbf{r}$ from the possibly backdoored multivariate Gaussian distribution $\mathcal{R}$, the server does not directly add the noise to the unperturbed query result $q\left(D\right)$. Instead, it rotates $\mathbf{r}$ to a direction $\mathbf{u}$ sampled uniformly at random, only preserving its length $\norm{\mathbf{r}}$. Equivalently, the noise added to the query result is $\sigma\cdot \norm{\mathbf{r}} \cdot \mathbf{u}$.

\begin{algorithm}[!htbp]
\caption{Noise-Rotated Gaussian Mechanism (NRGM)}\label{alg:mit-rotation}
\begin{algorithmic}[1]
\Require Database $D$; query $q$; variance $\sigma$; $d$-dimensional Gaussian distribution $\mathcal{R}$ (possibly replaced with hCLWE).
    \State $\mathbf{r}\sample \mathcal{R}$ \Comment{Sample the noise}
    \State $\mathbf{u}\sample \mathbb{S}^{d-1}$ \Comment{Sample a random direction}
    \State \Return $\mathbf{y}\gets q\left(D\right) +\sigma \cdot \norm{\mathbf{r}} \cdot \mathbf{u}$.
\end{algorithmic}
\end{algorithm}

Algorithm~\ref{alg:mit-rotation} does not affect the output distribution if the mechanism has not been backdoored, and can be executed in $\bigO{d}$ time and space complexity. However, if the mechanism is indeed backdoored as GPM, the equivalent distribution of the output $\mathbf{y}$ is $\mathcal{M}_{\sigma, \mathbf{w}', \beta, \gamma}$ for another uniformly random $\mathbf{w}'$ that is not known to the BPA. Therefore, the noise-rotated GPM maintains the same level of privacy guarantee as GM, as formalized in \appref{app:setup-regular-attacker-cdp}. Nevertheless, recall that the backdoor planted in the randomness of GPM can also be transformed into a backdoor in the randomness of sampling the random direction $\mathbf{u}$. Therefore, $\mathbf{r}$ and $\mathbf{u}$ should be sampled using different libraries. To prevent this cat-and-mouse game, we also consider adding further verification to fundamentally mitigate the existence of the backdoor.

\subsection{Verifiable DP}\label{sec:mit-vdp}

The integrity of GM may also be enforced by adding further verification to identify deviations from the authentic GM and unintended privacy leakages.

\subsubsection{Cryptographic Proofs} The execution of GM can be verified by zero-knowledge proofs (ZKPs), cryptographic primitives that certify the correct execution of prescribed computations~\cite{DBLP:journals/ftsec/Thaler22}. With ZKPs applied, an external verifier can be convinced that the authentic GM, rather than GPM, is genuinely executed by the server, such that no backdoor is planted in the server's output~\cite{vfuzz,dprio,DBLP:conf/iclr/ShamsabadiTCBHP24,KCY21,BC23,vddp}. However, the application of ZKPs incurs fundamental overhead, which poses significant challenges to their practical implementation.

\subsubsection{Formal Verifications} Formal verification techniques can identify bugs in the DP mechanisms implemented as computer programs, with the purpose of detecting unintended privacy leakages~\cite{DBLP:conf/vmcai/ReshefKSD24, DBLP:conf/ccs/WangDKZ20, DBLP:journals/corr/abs-1909-02481}. Such techniques can automatically identify the GPM backdoor as a breach of the intended privacy guarantee. However, this requires the open-sourcing of the DP libraries and the faithful execution of the verified programs.
\section{Experiments}\label{sec:exp}

In this section, we conduct experimental evaluations on GPM, with a focus on examining the power of distinguishing attacks gained by the BPA with knowledge of the backdoor key $\mathbf{w}$, as described in Section~\ref{sec:da}. All experiments were conducted on a computing node equipped with an NVIDIA A40 with 48GB VRAM, which shares 2 AMD EPYC 7272 12-Core CPUs and 512 GB RAM with 7 other computing nodes of the same hardware configuration. Unless specified otherwise, all experiments were repeated 100 times. The open-source code and raw experiment logs are available at \url{https://github.com/jvhs0706/GPM}.

\subsection{Experiments on Distinguishing Attacks}\label{sec:exp-da}

Since the covertness guarantees in Section~\ref{sec:setup-covertness} hinge on high dimensionality, it is only meaningful to consider additional privacy leakage under sufficiently large $d$ (e.g., $d \geq 256$) such that the original DP guarantees continue to hold against the RPA. Therefore, in this section, we report the success rates of distinguishing attacks; that is, the percentage of distinguishing attacks in which the index of the selected database is correctly identified by the BPA ($i=i^*$), which approximates the frequency in Equation~\eqref{eq:prob-da-success}. We conducted experiments for distinguishing attacks on two classes of queries that are ubiquitously studied in the DP literature where GM applies and that have a large output dimension: namely, DP histogram (DP-hist) queries~\cite{DBLP:conf/stoc/BassilyS15, DBLP:journals/vldb/XuZXYYW13, DBLP:conf/cikm/LiXJL15, DBLP:journals/pvldb/HayRMS10} and the gradient computation in DP stochastic gradient descent (DP-SGD)~\cite{DBLP:conf/ccs/AbadiCGMMT016, DBLP:conf/nips/HayesBM23, DBLP:conf/iclr/FangLF023}.

\begin{itemize}[leftmargin=*]
    \item For DP-hist queries, we experiment with artificial datasets $D \in \left\{1, 2, \dots, d \right\}^n$. That is, each of the $n$ records is assigned one of the $d$ classes. We set $d \in \left\{256, 4096, 65536\right\}$, such that the non-private query returns a $d$-dimensional vector $\mathbf{h}$ with the $i$-th element being $\mathbf{h}_i := \abs{\left\{D_j: D_j = i\right\}}$, i.e., the count of the $i$-th class in the database. Two databases $D$ and $D'$ are defined as neighbouring if $D$ is equivalent to $D'$ after adding or removing one record. Therefore, the histogram query has a sensitivity of $1$ under the $\ell^2$ norm.
    \item For the gradient computation in DP-SGD, we follow the algorithm of the original paper~\cite{DBLP:conf/ccs/AbadiCGMMT016} and the same setting as the distinguishing attacks in previous studies on DP numerical issues~\cite{DBLP:conf/sp/JinMRO22}. In particular, we experiment with the MNIST dataset using the LeNet-5 structure, and the CIFAR-10 dataset using VGG19, ResNet50, and MobileNet V2 structures, which have $2\times 10^6\leq d \leq 4\times 10^7$ parameters, respectively. Given a batch of $n=128$ data points $D=\left\{\left(\mathbf{x}_i, y_i\right)\right\}_{i=1}^n$, we compute the partial gradient of the loss function with respect to each data point $\mathbf{g}_i$, clip each $\mathbf{g}_i$ to an $\ell^2$ norm of $C\in \left\{4,8\right\}$ (denoting the clipped gradient as $\mathbf{g}_i'$), and define the non-private query result as $\mathbf{g}:=\frac{1}{n}\sum_{i=1}^n \mathbf{g}_i'$. The neighbouring relation is defined by replacing one data point in the batch. For fairness, we fix the random seeds for the training process to keep the non-replaced $\mathbf{g}_i'$s consistent for queries on the neighbouring databases. Therefore, the sensitivity of the query is $\frac{2C}{n}$ under the $\ell^2$ norm. We also consider two phases during training: before the model has converged (training from a randomly initialized set of parameters), and after the model has converged (loading a set of pretrained parameters with sufficiently high accuracy).
\end{itemize}

For both classes of queries, we fix $\delta^*=10^{-10}$ and vary $\epsilon^*$, where $\left(\epsilon^*, \delta^*\right)$ are the privacy parameters of the authentic GM as defined in Theorem~\ref{thm:gm-dp}, such that the scale of the noise $\sigma$ can be computed accordingly. Note that, by Theorems~\ref{thm:hclwe-ind-c} and \ref{thm:gpm-covert}, the computational indistinguishability result holds for any $\beta\left(d\right)\in d^{-O\left(1\right)}$ and $2\sqrt{d} \leq \gamma \leq d^{O\left(1\right)}$. Meanwhile, by Theorems~\ref{thm:gpm-dp-lb} and~\ref{thm:da}, $\beta$ is the decisive factor for the actual lower bound on privacy leakage and the success rates of the distinguishing attacks. Therefore, $\frac{1}{\beta}$ can be as large as \textbf{any polynomial}. Consequently, we choose $\beta$ in our experiments to range between $10^{-5}$ and $10^{-1}$ while varying $256\leq d \leq 4\times 10^7$, ensuring that $\frac{1}{\beta}$ is at most on the order of $d^2$, which satisfies the conditions. Similarly, we constrain $\gamma$ to a reasonable range between $2\sqrt{d}$ (the default value unless otherwise specified) and $20d$ in all experiments reported in this section.

It is also important to note that the purpose of the experiments on distinguishing attacks presented in this section is to showcase the seriousness of the GPM backdoor, and the evaluated attacks are independent of GM's numerical issues despite the similar experimental settings. In particular, the design and analysis of the distinguishing attack, as presented in Section~\ref{sec:da}, are fully principled in the real domain and therefore do not exploit vulnerabilities in floating-point arithmetic to increase the success rates.

\subsubsection{DP-hist} \label{sec:exp-da-hist}

\begin{figure*}[!t]
    \centering
    \includegraphics[width=\linewidth]{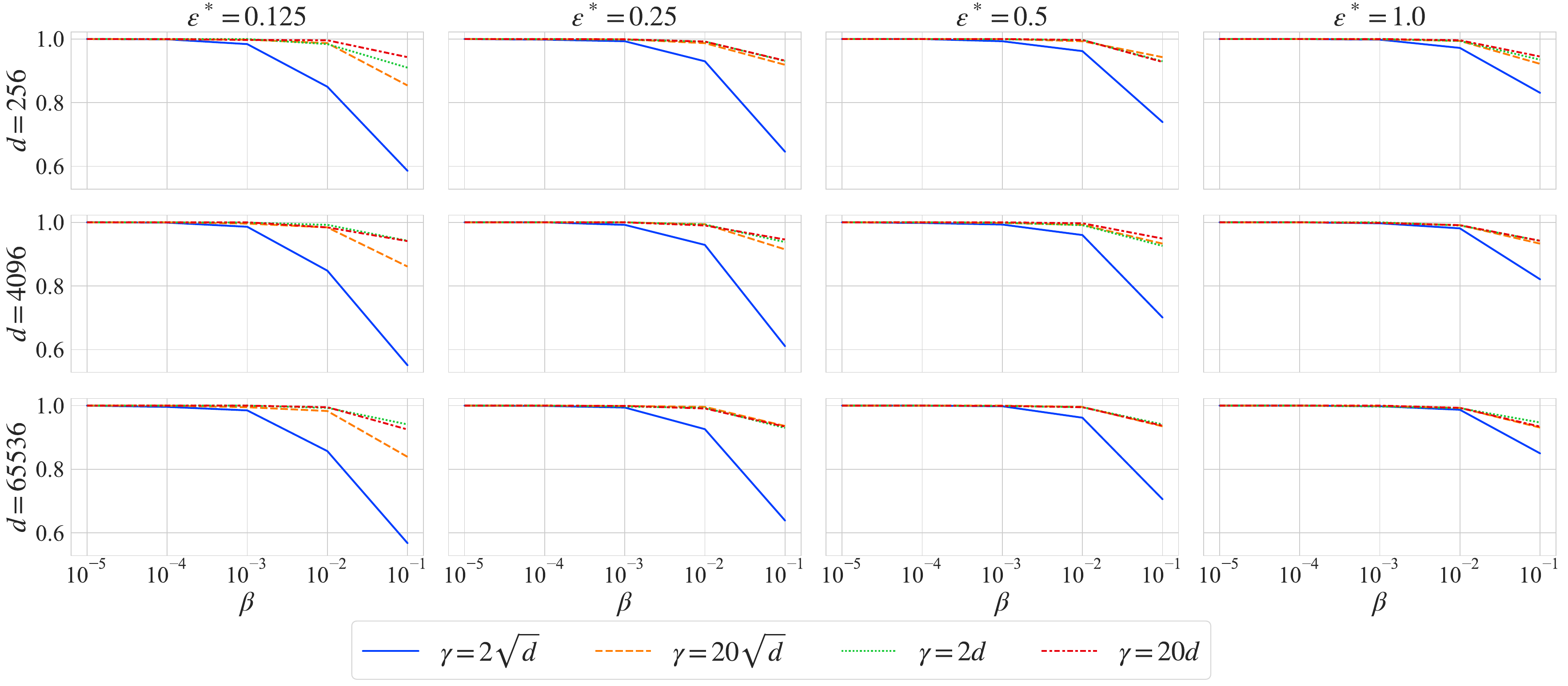}
    \caption{Distinguishing attack success rates on DP-hist. Success rates approach $1$ when $\beta < 10^{-3}$, and degrade significantly for $\beta \in \left\{10^{-1}, 10^{-2}\right\}$, especially for $\gamma = 2\sqrt{d}$. For larger $\beta$, increasing $\epsilon^*$ slightly improves performance.}
    \label{fig:hist-da}
\end{figure*}

As shown in Figure~\ref{fig:hist-da}, for DP-hist queries, we observe that when $\beta \leq 10^{-3}$, distinguishing attacks succeed with a probability close to 1. On the other hand, for $\beta = 0.1$, the success rates are not significantly higher than those of random guessing, especially for the smallest value of $\gamma = 2\sqrt{d}$. This phenomenon is in accordance with Theorem~\ref{thm:da} and is further visually explained in Figure~\ref{fig:gpm-density}, showing that the peaks, which are of width approximately $\sqrt{2\pi}\frac{\gamma}{\beta}$, are wider due to the larger values of $\beta$, resulting in greater overlaps in the probability masses. Therefore, the SCBP is motivated to select smaller values of $\beta$, as long as $\beta(d)$ is polynomially bounded as required by Theorem~\ref{thm:gpm-covert}.

Furthermore, it can also be observed that, for larger values of $\beta$, larger $\epsilon^*$ values generally result in better success rates. This phenomenon is in accordance with the $\left(\epsilon^*, \delta^*\right)$-DP in authentic GM, and can also be explained by the fact that larger $\epsilon^*$ values result in smaller $\sigma$ values. Consequently, for neighbouring databases $D$ and $D'$ and the unique pair of $T\in \Z$ and $t\in \left[-0.5, 0.5\right)$ such that $\sqrt{2\pi} \sigma \cdot (T + t) \frac{\gamma}{\beta^2 + \gamma^2} = \mathbf{w}^\top \left(q\left(D\right) - q\left(D'\right)\right)$, the probability that $T=0$ and $t \ll 1$ becomes lower over the randomness of sampling $\mathbf{w}$. By Theorem~\ref{thm:da}, this results in a higher success rate for distinguishing attacks. Similarly, note that for fixed $\sigma, D, D'$ and $\mathbf{w}$, the product of $\left(T+ t\right)$ and $\frac{\gamma}{\gamma^2+\beta^2}$ is constant. Consequently, the value of $\left(T+t\right)$ is approximately proportional to $\gamma$. Thus, for larger values of $\gamma$, $\left(T + t\right)$ generally takes on larger values; hence, the event where $T = 0$ and $t\ll 1$ becomes less likely. This explains why larger values of $\gamma$ compensate for the drop in success rates for larger values of $\beta$.

\subsubsection{DP-SGD}

\begin{figure*}[!t]
    \centering
    \includegraphics[width=\linewidth]{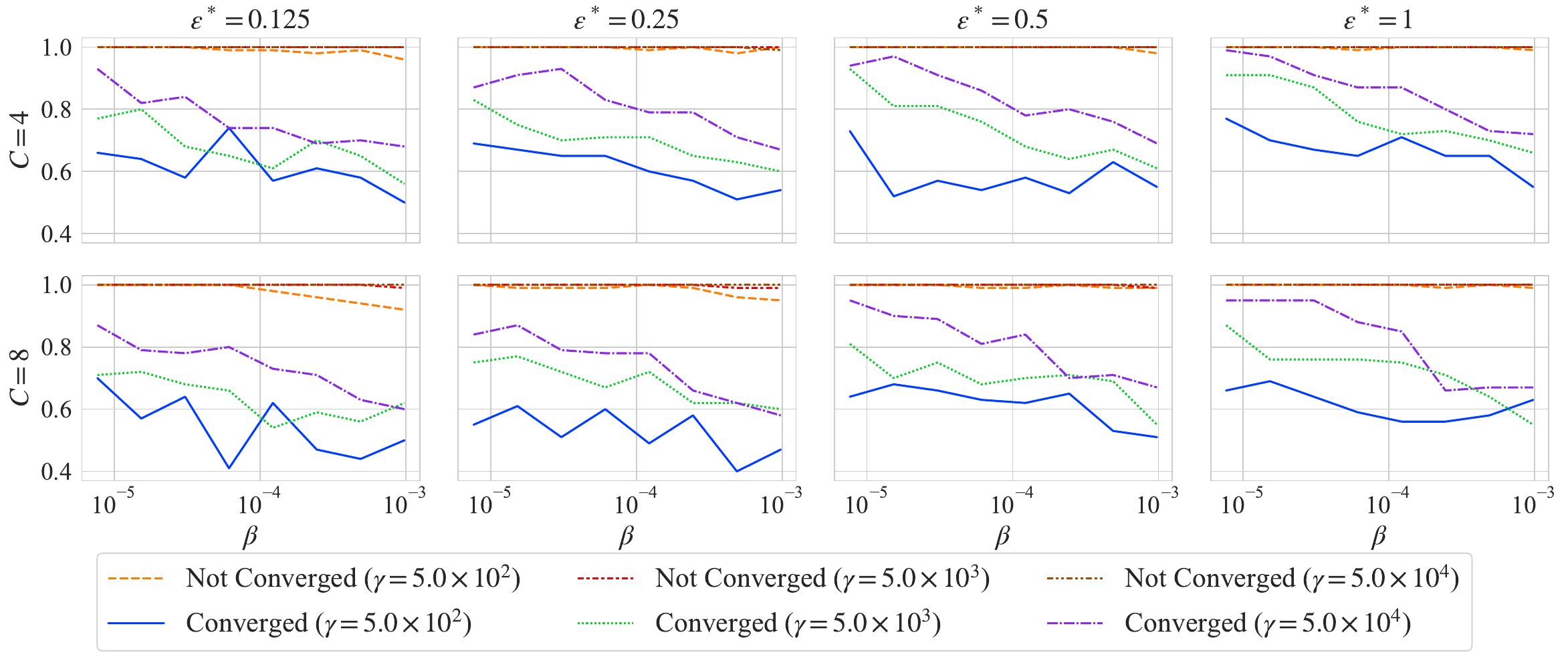}
    \caption{Distinguishing attack success rates on DP-SGD with MNIST. Success rates drop significantly after model convergence.}
    \label{fig:mnist-da}
\end{figure*}

We report the experimental results of distinguishing attacks on DP-SGD with the MNIST dataset and LeNet-5 in Figure~\ref{fig:mnist-da}. A trend similar to that in Figure~\ref{fig:hist-da} can be observed if the model is trained from random initialization. That is, the success rate of the distinguishing attacks approaches $1$ in all experiments when $\beta$ is at the level of $10^{-5}$ to $10^{-4}$, while at $\beta \approx 10^{-3}$, a slight drop in the success rate of 5\% to 10\% can be observed, and larger $\epsilon^*$ values, and therefore smaller $\sigma$ values, tend to compensate for this drop. However, one unexpected phenomenon can be observed on the SGD steps after the model has converged: the success rates are consistently below 1 even if smaller values of $\beta$ are chosen. This can be explained by the fact that, as the model has converged, the computed gradients afterwards are significantly smaller in magnitude~\cite{DBLP:conf/uss/ThudiJMSP24}. That is, $\norm{q\left(D\right) - q\left(D'\right)} \ll \Delta$, such that in Theorem~\ref{thm:da}, $T = 0$ and $t \ll 1$ with sufficiently high probability, which impacts the overall success rate of the distinguishing attacks. Meanwhile, similar to the DP-hist case, increasing the value of $\gamma$ counteracts the increase in this probability, thereby restoring the success rates after convergence to more than 90\% for larger $\epsilon^*$ values and smaller $\beta$ values.

\begin{table}[!t]
    \centering
    \caption{Distinguishing attack success rates on DP-SGD with CIFAR-10. Most instances achieve 100\% success rates. Unlike Figure~\ref{fig:mnist-da}, no drop in success rate is observed after convergence compared with before convergence. Note that, as detailed in \appref{app:additional-exp-gamma}, repeating the experiments with larger values of $\gamma = 20\sqrt{d}$ and $\gamma = 200\sqrt{d}$ (in contrast to the default value of $\gamma = 2\sqrt{d}$ in this table) results in 99--100\% success rates in all experimental instances.}
    \label{tab:cifar10-da}
    \autofit{\begin{tabular}{@{}lcccccccccc@{}}

\toprule
\multirow{2}{*}{Model} &
  \multirow{2}{*}{$d$} &
  \multirow{2}{*}{Converged} &
  \multirow{2}{*}{$\epsilon^*$} &
  \multirow{2}{*}{$\gamma (\times 10^3)$} &
  \multicolumn{3}{c}{$\beta\left(C=4\right)$} &
  \multicolumn{3}{c}{$\beta\left(C=8\right)$} \\
 &
   &
   &
   &
   &
  $10^{-5}$ &
  $10^{-4}$ &
  $10^{-3}$ &
  $10^{-5}$ &
  $10^{-4}$ &
  $10^{-3}$ \\ \midrule
\multirow{8}{*}{VGG19} &
  \multirow{8}{*}{39M} &
  \multirow{4}{*}{\ding{55}} &
  0.125 &
  \multirow{8}{*}{$12$} &
  100 &
  100 &
  99 &
  100 &
  99 &
  96 \\
 &  &                                     & 0.25  &  & 100 & 100 & 97  & 100 & 100 & 97  \\
 &  &                                     & 0.5   &  & 100 & 100 & 100 & 100 & 100 & 100 \\
 &  &                                     & 1     &  & 100 & 100 & 100 & 100 & 100 & 99  \\
 &  & \multirow{4}{*}{\ding{52} (94.0\%)} & 0.125 &  & 100 & 100 & 97  & 100 & 99  & 97  \\
 &  &                                     & 0.25  &  & 100 & 100 & 98  & 100 & 100 & 97  \\
 &  &                                     & 0.5   &  & 100 & 100 & 100 & 100 & 100 & 97  \\
 &  &                                     & 1     &  & 100 & 100 & 99  & 100 & 100 & 98  \\ \hline
\multirow{8}{*}{ResNet50} &
  \multirow{8}{*}{23M} &
  \multirow{4}{*}{\ding{55}} &
  0.125 &
  \multirow{8}{*}{$9.7$} &
  100 &
  100 &
  97 &
  100 &
  100 &
  97 \\
 &  &                                     & 0.25  &  & 100 & 100 & 97  & 100 & 100 & 97  \\
 &  &                                     & 0.5   &  & 100 & 100 & 100 & 100 & 100 & 100 \\
 &  &                                     & 1     &  & 100 & 100 & 100 & 100 & 100 & 100 \\
 &  & \multirow{4}{*}{\ding{52} (93.7\%)} & 0.125 &  & 100 & 100 & 99  & 100 & 100 & 98  \\
 &  &                                     & 0.25  &  & 100 & 100 & 98  & 100 & 100 & 99  \\
 &  &                                     & 0.5   &  & 100 & 100 & 100 & 100 & 100 & 99  \\
 &  &                                     & 1     &  & 100 & 100 & 100 & 100 & 100 & 100 \\ \hline
\multirow{8}{*}{MobileNetV2} &
  \multirow{8}{*}{2.2M} &
  \multirow{4}{*}{\ding{55}} &
  0.125 &
  \multirow{8}{*}{$3.0$} &
  100 &
  98 &
  97 &
  100 &
  100 &
  98 \\
 &  &                                     & 0.25  &  & 100 & 100 & 98  & 100 & 100 & 100 \\
 &  &                                     & 0.5   &  & 100 & 100 & 100 & 100 & 99  & 100 \\
 &  &                                     & 1     &  & 100 & 100 & 97  & 100 & 100 & 100 \\
 &  & \multirow{4}{*}{\ding{52} (93.9\%)} & 0.125 &  & 100 & 100 & 98  & 100 & 99  & 93  \\
 &  &                                     & 0.25  &  & 100 & 100 & 100 & 100 & 100 & 100 \\
 &  &                                     & 0.5   &  & 100 & 100 & 99  & 100 & 100 & 98  \\
 &  &                                     & 1     &  & 100 & 100 & 100 & 100 & 100 & 100 \\ \bottomrule
\end{tabular}}
\end{table}

\begin{table}[!t]
    \centering
    \caption{Median magnitude of gradients and the difference between gradients computed on neighbouring databases, both of which are significantly larger in CIFAR-10 tasks regardless of the convergence status.}
    \label{tab:gradient-magnitude}
    \autofit{\begin{tabular}{cccccc}
\toprule
\multirow{2}{*}{Task} & Model       & \multicolumn{2}{c}{$\norm{q\left(D\right)}$} & \multicolumn{2}{c}{$\norm{q\left(D\right)-q\left(D'\right)}$} \\
                      & Converged   & \ding{55}             & \ding{52}            & \ding{55}                & \ding{52}                          \\ \midrule
MNIST                 & LeNet5      & 0.32                  & 0.025                & 0.32                     & $3.3\times 10^{-6}$                \\
CIFAR-10              & VGG19       & 12                    & 1.4                  & 13                       & 0.81                               \\
CIFAR-10              & ResNet50    & 980                   & 2.9                  & 110                      & 0.31                               \\
CIFAR-10              & MobileNetV2 & 11                    & 6.2                  & 1.35                     & 0.59                               \\ \bottomrule
\end{tabular}}
\end{table}

We further report the experimental results on CIFAR-10 in Table~\ref{tab:cifar10-da}. It can be observed that 100\% success rates are achieved for all instances when $\beta = 10^{-5}$, and for most instances when $\beta = 10^{-4}$. A slight drop in success rate is observed when $\beta = 10^{-3}$, although it still remains at or above 97\% in most cases. Moreover, no drop in success rate is observed after convergence. This phenomenon can be explained by the fact that the gradients and their differences remain of considerable magnitude, as illustrated in Table~\ref{tab:gradient-magnitude}. This is attributed to the relatively lower accuracies on CIFAR-10, which result in loss functions that are not very close to $0$, as well as the larger model sizes, which increase the magnitudes of the gradients. Consequently, the edge case in Theorem~\ref{thm:da}, where $T = 0$ and $t \ll 1$, no longer occurs with significant probability.

\subsection{Additional Experiments}\label{sec:additional-exp}

\begin{table}[!t]
    \centering
    \caption{Comparison between GM's and GPM's actual $\ell^2$ errors and the theoretical expectation computed from GM. No significant differences are observed.}
    \autofit{\begin{tabular}{cccccc}
    \toprule
                               & $\epsilon^*$ & 0.125   & 0.25   & 0.5    & 1      \\ \midrule
    \multirow{3}{*}{$d=256$}   & Expect       & 815.5   & 408.4  & 204.8  & 103.2  \\
                               & GM           & 811.8   & 404.9  & 203.9  & 102.6  \\
                               & GPM          & 819.1   & 407.8  & 206.1  & 102.2  \\ \hline
    \multirow{3}{*}{$d=4096$}  & Expect       & 3262.0  & 1633.5 & 819.3  & 412.1  \\
                               & GM           & 3263.3  & 1636.1 & 819.9  & 411.4  \\
                               & GPM          & 3260.3  & 1632.9 & 819.3  & 412.5  \\ \hline
    \multirow{3}{*}{$d=65536$} & Expect       & 13048.1 & 6534.1 & 3277.0 & 1648.4 \\
                               & GM           & 13043.7 & 6536.5 & 3277.2 & 1648.7 \\
                               & GPM          & 13044.2 & 6526.8 & 3276.3 & 1648.3 \\ \bottomrule
\end{tabular}}
    \label{tab:l2}
\end{table}

By Theorem~\ref{thm:gpm-covert}, the GPM backdoor cannot be discovered by any PPT adversary without knowledge of the backdoor key, except with negligible probability. That is, it is impossible to distinguish the outputs of GM and GPM using any metric independent of $\mathbf{w}$, provided the metric can be computed efficiently. To further verify this property, we measure the actual $\ell^2$ error, averaging over all instances of the experiments considered in Figure~\ref{fig:hist-da}, and present the results in Table~\ref{tab:l2}. It can be observed that the actual $\ell^2$ error and the theoretically computed $\ell^2$ error are not significantly different. This closeness in error magnitude serves as empirical evidence of the backdoor's undiscoverability.

\begin{figure}
    \centering
    \includegraphics[width=0.8\linewidth]{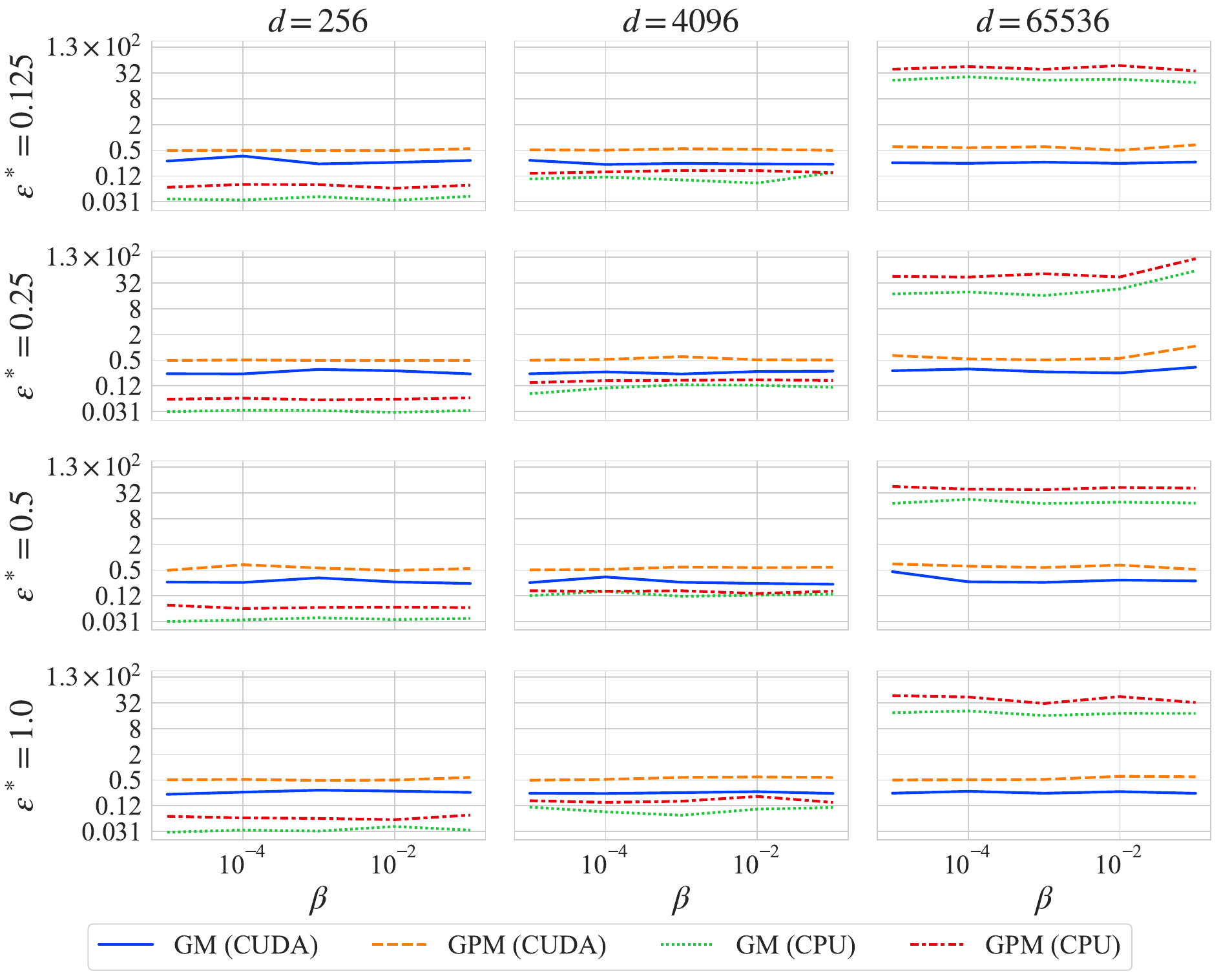}
    \caption{Comparison of noise sampling times between GM and GPM (in milliseconds), measured on CPU and CUDA devices; GPM incurs a 0.5--1.5x slowdown. Note that GPM's implementation is not fully optimized. The running times of the other components (e.g., the unperturbed query processing time) are identical between GM and GPM, yet not included.}
    \label{fig:hist-rt}
\end{figure}

While we acknowledge the challenge of exactly matching the execution times of GM and GPM, we argue that the differences in running times are unlikely to expose the backdoor. In particular, Figure~\ref{fig:hist-rt} compares the noise sampling times of GM and GPM under various DP-hist configurations; the results show that Gaussian pancake noise is sampled only 0.5--1.5x slower than standard Gaussian noise. We further note that this difference is a \textbf{pessimistic estimate} of GPM's covertness regarding this side channel. Specifically:

\begin{itemize}[leftmargin=*]
    \item Under the threat model described in Section~\ref{sec:setup-threat-model}, where the closed-source benign implementation is replaced by a malicious backdoored version in the software supply chain, the server lacks access to the benign software as a reference. Similarly, the server cannot obtain the exact implementation details of either the benign or the compromised software to derive an accurate estimate of execution time. Therefore, the 0.5--1.5x slowdown, which remains within the same order of magnitude, could be plausibly attributed to implementation variances or uncontrollable factors such as hardware load and cooling.
    \item The reported running time corresponds \textbf{only} to noise sampling, which is the sole source of the execution-time difference between GM and GPM. All running times other than noise sampling (e.g., the unperturbed query processing time) are excluded from Figure~\ref{fig:hist-rt} and are identical for GM and GPM. In particular, differences in execution times would be diluted by the unperturbed query execution and any applicable post-processing. Consequently, we anticipate that the differences in end-to-end execution times between the two mechanisms will be less significant. 
\end{itemize}

Also, unlike the implementation of GM, which enjoys low-level hardware accelerations, our implementation of GPM is not fully optimized at an industrial level. Specifically, GM was implemented using industrial-level APIs such as \texttt{torch.randn}, which have undergone hardware-level optimizations (e.g., compression into a single CUDA kernel and instruction set optimizations), whereas GPM relies on higher-level tensor operations (typically involving multiple CUDA kernel calls). Therefore, it is foreseeable that dedicated SCBPs could optimize GPM to achieve an even greater level of covertness with a smaller difference in running times.
\section{Related Work}

Numerical issues have historically been among the most common and exploitable software defects undermining the rigorous theoretical guarantees provided by differential privacy. In particular, though DP mechanisms are mathematically sound when formulated in the real-number domain, their floating-point implementations can subtly yet damagingly deviate from their original design, typically due to imprecisions in the least significant digits. Mironov~\cite{DBLP:conf/ccs/Mironov12} was the first to identify that floating-point imprecision in the Laplace mechanism led to unintended leakage, as the actual distribution of the least significant bits highly depends on the unperturbed query result and therefore leaks significantly more information about it than intended. As a result, privacy attacks were found to be successful against all four publicly available libraries evaluated. As a countermeasure, it was proposed to post-process the least significant bits to restore the intended DP guarantees. Similarly, Ilvento~\cite{DBLP:conf/ccs/Ilvento20} demonstrated related risks in the Exponential mechanism, where rounding errors can result in some outputs being assigned zero probability, effectively setting the privacy cost to $\epsilon = +\infty$. To address this, a base-2 variant was proposed to better align with hardware-level arithmetic in binary computers. More generally, Casacuberta et al.~\cite{DBLP:conf/ccs/CasacubertaSVW22} showed that such vulnerabilities are more widespread than previously believed, particularly in basic operations like summation, where arithmetic discrepancies can cause the true sensitivity to be underestimated and insufficient noise to be added. These results underscore the importance of secure and precise implementations, even for theoretically well-understood DP mechanisms.

Most relevant to our work, Jin et al.~\cite{DBLP:conf/sp/JinMRO22} found similar vulnerabilities in the Gaussian mechanism across multiple standard sampling algorithms used in the implementations of the Gaussian mechanism~\cite{Ziggurat,Marsaglia,DBLP:journals/tc/LeeVLL06}, in addition to previously studied timing attacks~\cite{DBLP:conf/uss/HaeberlenPN11, DBLP:conf/ccs/RatliffV24, DBLP:conf/sp/AndryscoKMJLS15}. In particular, similar to the Laplace mechanism, the least significant bits are also a source of unintended privacy leakage. Leveraging this vulnerability, they proposed distinguishing attacks that successfully exploited this software defect, which caused only a minor deviation from the theoretical design under which the originally sound DP guarantee holds. In response, patches have been introduced in IBM's differential privacy library~\cite{diffprivlib, DBLP:conf/ccs/HolohanB024} to obfuscate the least significant bits, which were shown to be especially vulnerable.

Goldwasser et al.~\cite{DBLP:conf/focs/GoldwasserKVZ22} presented the first application of the CLWE problem~\cite{clwe} for constructing undetectable backdoors in deep learning models. Their work introduced several cryptographically motivated backdoor mechanisms, including those based on digital signatures~\cite{DBLP:conf/stoc/GoldwasserMR85}, sparse PCA~\cite{DBLP:conf/colt/BerthetR13}, and CLWE. In the CLWE-based construction, the authors considered Gaussian features extracted from a randomly initialized fully connected layer that approximates the Gaussian kernel~\cite{DBLP:conf/nips/RahimiR07}. These features are then used to train a binary linear classifier. The backdoor is planted by replacing the Gaussian distribution with a CLWE distribution during initialization, such that modifying the input along a secret vector $\mathbf{w}$ flips the classifier's prediction. This construction highlights the potential for CLWE to be applied beyond learning theory~\cite{DBLP:conf/icml/DiakonikolasKR23, DBLP:conf/colt/DiakonikolasKPZ23, DBLP:conf/nips/DiakonikolasKRS23, DBLP:conf/iclr/ChenC0LSZ23, DBLP:conf/nips/ShahCK23}.

Formal methods for verifying differential privacy have evolved significantly from foundational work on type systems for sensitivity analysis~\cite{DBLP:conf/icfp/ReedP10}. These systems introduced a functional language in which types track function sensitivity, enabling automatic proofs of differential privacy. Building on this foundation, Duet~\cite{DBLP:journals/corr/abs-1909-02481} proposed an expressive higher-order language and linear type system for verifying the privacy of general-purpose higher-order programs. CheckDP~\cite{DBLP:conf/ccs/WangDKZ20} advanced automated verification by introducing a unified framework that can both prove and disprove differential privacy claims. It is capable of finding subtle bugs and generating counterexamples that elude prior techniques. More recently, Reshef et al.~\cite{DBLP:conf/vmcai/ReshefKSD24} introduced local differential classification privacy (LDCP), a new privacy notion designed for black-box classifiers, extending the idea of local robustness into the DP domain. These efforts have inspired ongoing work in the formal auditing of modern DP systems~\cite{DBLP:conf/ccs/LoknaPDV23, DBLP:conf/nips/PillutlaAKMOO23} in practical deployments.

Beyond formal verification of DP programs, a growing body of work has focused on making the execution of DP mechanisms themselves verifiable. This includes proof-carrying implementations of DP mechanisms in the Fuzz programming language~\cite{vfuzz}, as well as verifiable instantiations of randomized response~\cite{KCY21} and the floating-point Gaussian mechanism used in DP-SGD~\cite{DBLP:conf/iclr/ShamsabadiTCBHP24}. With the increasing reliance on distributed data analysis and secure multiparty computation, verifiability has also been explored in collaborative settings. Notable examples include verifiable DP protocols for Prio~\cite{dprio}, the Binomial mechanism~\cite{BC23}, and the discrete Laplace mechanism~\cite{vddp}.
\section{Conclusion} \label{sec:concl}

This study introduces the Gaussian pancake mechanism (GPM), a novel and practical construction of a backdoor attack against the Gaussian mechanism (GM), one of the most prevalent DP mechanisms in industrial applications. We have shown that, for a backdoor attacker with knowledge of the backdoor key, the privacy guarantee degrades arbitrarily from both theoretical and empirical perspectives, while the attack remains undetectable by the server and third parties. In light of these findings, we advocate for the use of open-source DP libraries, formal verification techniques, and cryptographic safeguards to mitigate such risks. Future research may further explore stronger yet more elusive adversarial manipulations of DP mechanisms that \textbf{1)} could result in exacerbated privacy loss (e.g., full recovery of the unperturbed query result $q\left(D\right)$, the input database $D$, or an end-to-end privacy attack against DP-SGD), and \textbf{2)} are strictly computationally undetectable, not only via checking the mechanism's outputs but also through side channels such as execution times. The purpose of developing such attacks should not be to weaken privacy, but to advance the robustness, verifiability, and trustworthiness of privacy-preserving technologies. To this end, we also encourage future work on more robust countermeasures that extend the noise rotation technique to reduce reliance on trusted external randomness, as well as the design of more efficient verification protocols to certify the absence of backdoors in practical, large-scale deployments.

\bibliographystyle{IEEEtran}
\bibliography{reference}

\ifdefined\includeappendix
    \appendix 
    \section{Further Details on Problem Setup}\label{app:setup}

\subsection{Computational DP Guarantee Against the RPA} \label{app:setup-regular-attacker-cdp}

To the RPA without knowledge of the secret $\mathbf{w}$, the privacy guarantee of GPM can be captured using \textbf{computational differential privacy}, as formalized in Definition~\ref{def-app:ind-cdp}. Compared with the (statistical) DP guarantee, which holds against computationally unbounded adversaries, Definition~\ref{def-app:ind-cdp} only holds against PPT adversaries. Here, we use $\kappa$ as the generic security parameter. 

\begin{Def}[IND-CDP~\cite{DBLP:conf/crypto/MironovPRV09}]\label{def-app:ind-cdp}
    Given $\epsilon: \N \to \R_+$, a family of mechanisms $\left\{\M_\kappa: \D_\kappa \to \Y_\kappa \right\}_{\kappa\in \N}$ is \emph{$\epsilon$}-IND-CDP if and only if for any family of PPT algorithms $\left\{\adv_\kappa\right\}_{\kappa\in \N}$, there exists $\mu(\kappa)\in \negl[\kappa]$ such that for any pair of neighbouring databases $D_\kappa\simeq D_\kappa'$, \begin{equation}
        \Pr\left[\adv_\kappa\left(\M_\kappa\left(D_\kappa\right)\right)=1\right]\leq e^{\epsilon(\kappa)}\Pr\left[\adv_\kappa\left(\M_\kappa\left(D_\kappa'\right)\right)=1\right] + \mu(\kappa).
    \end{equation}
\end{Def}

To facilitate the security analysis, we present in Definition~\ref{def-app:dp-negl} an extended version of differential privacy, involving a family of mechanisms and negligible $\delta$. 

\begin{Def}[differential privacy (extended)~\cite{DBLP:journals/fttcs/DworkR14, DBLP:conf/tcc/DworkMNS06, DBLP:conf/crypto/MironovPRV09, DBLP:conf/tcc/GroceKY11}]\label{def-app:dp-negl}
    Given $\epsilon: \N \to \R_+$ and $\delta: \N\to\left[0, 1\right]$, a family of mechanisms $\left\{\M_\kappa: \D_\kappa \to \Y_\kappa \right\}_{\kappa\in \N}$ is \emph{$\left(\epsilon, \delta\right)$-differentially private} if, for any pair of neighbouring databases $D_\kappa\simeq D_\kappa'$ that differ in a single row, and any measurable set $S_\kappa\subseteq \Y_\kappa$, 
    \begin{equation}
        \prob{\mathcal{M}_\kappa\left(D_\kappa\right) \in S_\kappa}\leq e^{\epsilon(\kappa)}\prob{\M_\kappa\left(D_\kappa'\right) \in S_\kappa} + \delta(\kappa).
    \end{equation} 
    If $\delta(\kappa)\in \negl[\kappa]$, we call $\left\{\M_\kappa\right\}_{\kappa\in \N}$ $\left(\epsilon, \sf{negl}\right)$-DP.
\end{Def}

GM also satisfies $\left(\epsilon, \sf{negl}\right)$-DP under Definition~\ref{def-app:dp-negl}, as formalized in Proposition~\ref{prop-app:gm-dp-negl}.

\begin{Prop}\label{prop-app:gm-dp-negl}
    Consider a family of Gaussian mechanisms $\left\{\M_\kappa: \D_\kappa \to \R^{d_\kappa}\right\}_{\kappa\in \N}$ with sensitivities $\left\{\Delta_\kappa\right\}_{\kappa\in \N}$ and variances $\left\{\sigma_\kappa^2I_{d_\kappa}\right\}_{\kappa\in \N}$. With $\delta(\kappa)=2^{-\kappa}$, $\left\{\M_\kappa\right\}_{\kappa\in \N}$ is $\left(\epsilon, \sf{negl}\right)$-DP, where $\epsilon(\kappa)\in \frac{\Delta_\kappa^2}{2\sigma_\kappa^2} + \frac{\Delta_\kappa}{\sigma_\kappa}\bigTheta{\sqrt{\kappa}}$.
\end{Prop}

To the RPA who does not know the value of the sampled $\mathbf{w}$, GPM satisfies $\epsilon$-IND-CDP with the same $\epsilon$ as Proposition~\ref{prop-app:gm-dp-negl}. That is, GPM and GM achieve the same level of privacy guarantee against the RPA, as the RPA does not have any non-negligible advantage in computing extra information from GPM's output compared to GM's output. 

\begin{Lem}\label{lem-app:gpm-ind-cdp}
    Given a family of queries $\left\{q_d: \D_d \to \R^d\right\}_{d\in \N}$ with sensitivities $\left\{\Delta_d\right\}_{d\in \N}$, consider the corresponding family of GPMs $\left\{\mathcal{M}_{\sigma_d, \mathbf{w}_d, \beta(d), \gamma(d)}: \D_d \to \R^d\right\}_{d\in \N}$. If $2\sqrt{d}\leq \gamma(d) \leq d^{O(1)}$ and $\beta(d)= d^{-O(1)}$, then $\left\{\mathcal{M}_{\sigma_d, \mathbf{w}_d, \beta(d), \gamma(d)}\right\}_{d\in \N}$ is $\epsilon$-IND-CDP with $\epsilon(d)\in \frac{\Delta_d^2}{2\sigma_d^2} + \frac{\Delta_d}{\sigma_d}\bigTheta{\sqrt{d}}$, under the assumption that $\mathbf{SIVP}\notin \sf{BQP}$ or $\mathbf{GapSVP}\notin \sf{BQP}$ \cite{clwe,DBLP:conf/focs/GoldwasserKVZ22}, for uniformly randomly sampled $\mathbf{w}_d\sample \mathbb{S}^{d-1}$.
\end{Lem}

\begin{proof}[Proof of Lemma~\ref{lem-app:gpm-ind-cdp}]
    By Theorem~\ref{thm:gpm-covert}, given any sequence of input databases $\left\{D_\kappa\right\}_{\kappa\geq 0}$, $\left\{\mathcal{M}_\kappa\left(D_\kappa\right)\right\}$ and $\left\{\mathcal{M}_\kappa\left(D_\kappa'\right)\right\}$ are computationally indistinguishable. Therefore, by Mironov et al.~\cite{DBLP:conf/crypto/MironovPRV09} and Groce et al.~\cite{DBLP:conf/tcc/GroceKY11}, GPM satisfies $\epsilon$-SIM\textsubscript{$\forall\exists$}-CDP, and therefore $\epsilon$-IND-CDP.
\end{proof}

\subsection{Deferred Proofs}

\begin{proof}[Proof of Theorem~\ref{thm:gpm-covert}]
    Assume the opposite, i.e., there exists a sequence $\left(D_{d, i}\right)_{d\in \N, 1\leq i \leq n(d)}$ such that the difference in Inequality~\eqref{eq:gpm-covert} is lower bounded by $\mu^*(d)\not\in\negl[\lambda]$. Therefore, the adversary $\adv_d'\left(\mathbf{r}_1, \mathbf{r}_2, \dots, \mathbf{r}_{n(d)}\right)$ (Figure~\ref{fig:gpm-covert-proof-adv}) can efficiently distinguish between samples from $\mathcal{N}\left(\mathbf{0}, I_d\right)$ and $\sqrt{2\pi}\mathcal{H}_{\mathbf{w}_d, \beta(d), \gamma(d)}$, i.e., \begin{multline}
        \left|\prob{\adv_d'\left(\mathbf{r}_1, \mathbf{r}_2, \dots, \mathbf{r}_{n(d)}\right):\\\mathbf{r}_1, \mathbf{r_2}, \dots, \mathbf{r}_n\sample \mathcal{N}\left(\mathbf{0}, I_d\right)}-\prob{\adv_d'\left(\mathbf{r}_1, \mathbf{r}_2, \dots, \mathbf{r}_{n(d)}\right):\\\mathbf{w}_d\sample \mathbb{S}^{d-1}\\\mathbf{r}_1, \mathbf{r_2}, \dots, \mathbf{r}_n\sample \sqrt{2\pi}\mathcal{H}_{\mathbf{w}_d, \beta(d), \gamma(d)}}\right|\geq \mu^*\left(d\right),
    \end{multline} which violates the hardness assumption of the hCLWE problem.

    \begin{figure}
        \centering
        \procedureblock[space=auto,linenumbering]{$\adv_d'\left(\mathbf{r}_1, \mathbf{r}_2, \dots, \mathbf{r}_{n(d)}\right)$}{\pcfor 1\leq i \leq n(d)\\
           \mathbf{y}_i \gets q\left(D_{d,i}\right) + \sigma(d)\mathbf{r}_i\\
        \pcendfor\\
        \pcreturn \adv_d\left(\mathbf{y}_1, \mathbf{y}_2, \dots, \mathbf{y}_{n(d)}\right)}
        \caption{Definition of $\adv_d'\left(\mathbf{r}_1, \mathbf{r}_2, \dots, \mathbf{r}_{n(d)}\right)$.}
        \label{fig:gpm-covert-proof-adv}
    \end{figure}
\end{proof}
    \section{Further Privacy Analysis on GPM}\label{app:privacy}

\subsection{Proof of Key Lemma} \label{app:privacy-one-dimensional-hclwe-diff}

\begin{proof}[Proof of Lemma~\ref{lem:one-dimensional-hclwe-diff}]
    Observe that for any pair $z\neq z'$, $I_z(t)\cap I_{z'}(t) = \emptyset$. Therefore, by abbreviating $a_{\beta, \gamma, z}$ as $a_z$,
    {\allowdisplaybreaks\begin{align}
        ~ & \prob{Y \in A(t)}\\
        = & \sum_{z\in \Z}\int_{I_z(t)}\psi_{1, \beta, \gamma}(y)\,dy\\
        \geq & \sum_{z\in \Z}\int_{I_z(t)}a_z\phi\left(\mathbf{y}; \frac{\gamma z}{\beta^2 + \gamma^2}, \frac{1}{2\pi}\cdot\frac{\beta^2}{\beta^2 + \gamma^2}\right)\,dy\\
        = & \sum_{z\in \Z} a_z \int_{\frac{\gamma \left(z-\frac{\abs{t}}{2}\right)}{\beta^2 + \gamma^2}}^{\frac{\gamma \left(z+\frac{\abs{t}}{2}\right)}{\beta^2 + \gamma^2}}\phi\left(\mathbf{y}; \frac{\gamma z}{\beta^2 + \gamma^2}, \frac{1}{2\pi}\cdot\frac{\beta^2}{\beta^2 + \gamma^2}\right)\,dy\\
        = & \sum_{z\in \Z}a_z \left(1 - 2\Phi\left(-\frac{\gamma\abs{t}}{\beta}\sqrt{\frac{\pi}{2\left(\beta^2 + \gamma^2\right)}}\right)\right)\\
        = & 1 - 2\Phi\left(-\frac{\gamma\abs{t}}{\beta}\sqrt{\frac{\pi}{2\left(\beta^2 + \gamma^2\right)}}\right). \label{eq:hclwe-peak}
    \end{align}}

    Additionally, observe that for any $T \in \Z$ and $t \in [-0.5, 0) \cup (0, 0.5)$, $A(t) - T - t \cap A(t) = \emptyset$. Therefore,
    \begin{align}
        \prob{Y' \in A(t)} =& \prob{Y \in A(t) - T - t}\\
        \leq & 1 - \prob{Y \in A(t)}\\
        \leq & 2\Phi\left(-\frac{\gamma\abs{t}}{\beta}\sqrt{\frac{\pi}{2\left(\beta^2 + \gamma^2\right)}}\right),
    \end{align}
    completing the proof.
\end{proof}

\subsection{Further Details on the Lower Bound} \label{app:privacy-lb}

\begin{proof}[Proof of Theorem~\ref{thm:gpm-dp-lb}]
    Let $\mathbf{Y} \sim \mathcal{M}_{\sigma, \mathbf{w}, \beta, \gamma}\left(D\right)$ and $\mathbf{Y}' \sim \mathcal{M}_{\sigma, \mathbf{w}, \beta, \gamma}\left(D'\right)$ be two multivariate random variables. By projecting onto $\mathbf{w}$,
    \begin{align}
        \frac{\left(\mathbf{Y} - q(D)\right)^\top \mathbf{w}}{\sqrt{2\pi}\sigma} &\sim \mathcal{H}_{1, \beta, \gamma},\\
        \frac{\left(\mathbf{Y}' - q(D)\right)^\top \mathbf{w}}{\sqrt{2\pi}\sigma} &\sim \mathcal{H}_{1, \beta, \gamma} + \left(T + t\right)\frac{\gamma}{\beta^2 + \gamma^2}.
    \end{align}

    Using the same notation as in Lemma~\ref{lem:one-dimensional-hclwe-diff}, define
    \begin{equation}
        S(D, t) := \bigcup_{z\in \Z} \left\{ \mathbf{y} :\frac{\left(\mathbf{y} - q(D)\right)^\top \mathbf{w}}{\sqrt{2\pi} \sigma}\in I_z(t)\right\},\label{eq:set-prob-mass-concentrated}
    \end{equation}
    then by Lemma~\ref{lem:one-dimensional-hclwe-diff},
    \begin{align}
        \prob{\mathbf{Y}\in S\left(D, t\right)} \geq& 1 - 2\Phi\left(-\frac{\gamma\abs{t}}{\beta}\sqrt{\frac{\pi}{2\left(\beta^2 + \gamma^2\right)}}\right),\\
        \prob{\mathbf{Y}'\in S\left(D, t\right)} \leq& 2\Phi\left(-\frac{\gamma\abs{t}}{\beta}\sqrt{\frac{\pi}{2\left(\beta^2 + \gamma^2\right)}}\right).
    \end{align}
    Therefore, if $\prob{\mathbf{Y}\in S\left(D, t\right)} \leq e^{\epsilon} \prob{\mathbf{Y}'\in S\left(D, t\right)}+\delta$, Inequality \eqref{eq:gpm-dp-lb} cannot hold.
\end{proof}

We rely on Lemma~\ref{lem-app:sqrt-beta}~\cite{DBLP:journals/jmlr/CaiFJ13} to analyze the lower bound of GPM's privacy cost under the randomness of $\mathbf{w}$.

\begin{Lem}\label{lem-app:sqrt-beta}
    For fixed $\mathbf{u}\in \mathbb{S}^{d-1}$ and random variable $\mathbf{V}\sim \mathbb{S}^{d-1}$, the distribution of $T := \mathbf{u}^\top\mathbf{V}$ (denoted $\mathcal{P}_t$) has p.d.f. \begin{equation}
        f_T(t) = \frac{\Gamma\left(\frac{d}{2}\right)}{\sqrt{\pi} \Gamma\left(\frac{d-1}{2}\right)} \left(1 - t^2\right)^{\frac{d-3}{2}},
    \end{equation}
    such that $\mathbb{E}(T) = 0$ and $\text{Var}(T) = \frac{1}{d}$.~\cite{DBLP:journals/jmlr/CaiFJ13}
\end{Lem} 

\begin{proof}[Proof of Proposition~\ref{prop:gpm-dp-lb}]
    By definition, there exists a pair of neighbouring databases $D\simeq D'$ such that $\norm{q\left(D\right) - q\left(D'\right)} = \Delta$. By Lemma~\ref{lem-app:sqrt-beta}, for randomly sampled $\mathbf{w} \sample \mathbb{S}^{d-1}$,
    \begin{equation}
        \mathbf{w}^\top (q(D') - q(D)) \sim \Delta \cdot \mathcal{P}_T,
    \end{equation}
    which has range $[-\Delta, \Delta]$, mean $0$, and variance $\frac{\Delta^2}{d}$.

    Given $\gamma \gg \sqrt{d} \cdot \frac{\sigma}{\Delta}$ and $\beta \ll 1$, we have $\frac{\Delta}{\sqrt{d}} \gg \sqrt{2\pi} \sigma \cdot \frac{\gamma}{\beta^2 + \gamma^2}$. By Weyl's Equidistribution Theorem,
    \begin{equation}
        \Pr_{\mathbf{w} \sample \mathbb{S}^{d-1}} \left[ \abs{t} \geq 0.25 \right] \in 0.5 + o(1),
    \end{equation}
    in which case $\mathcal{M}_{\sigma, \mathbf{w}, \beta, \gamma}$ is not $\left(\epsilon, \delta\right)$-DP for
    \begin{equation}
        \epsilon < \log\left( \frac{1-\delta}{2 \Phi\left( -\frac{\gamma}{4\beta} \sqrt{\frac{\pi}{2(\beta^2 + \gamma^2)}} \right)} - 1 \right) \in \Theta\left( \frac{1}{\beta^2} \right)
    \end{equation} and any $0 < \delta \leq 0.5$.
\end{proof}

\subsection{Further Details on the Upper Bound} \label{app:privacy-ub}

\begin{proof}[Proof of Theorem~\ref{thm:gpm-dp-ub}]
    Consider 
    \begin{equation}
        \mathcal{G}_{\sigma, \mathbf{w}, \beta, \gamma}(D) := q(D) + \sqrt{2\pi} \sigma \cdot \mathcal{N}\left(\mathbf{0}, \Sigma^*_{\mathbf{w}, \beta, \gamma}\right),
    \end{equation} 
    for any pair of neighbouring databases $D\simeq D'$. The distributions of $\mathcal{G}_{\sigma, \mathbf{w}, \beta, \gamma}(D)$ and $\mathcal{G}_{\sigma, \mathbf{w}, \beta, \gamma}(D')$ are
    \begin{align}
        \mathcal{N}\left(q(D), \sigma^2\left(I - \frac{\gamma^2}{\beta^2 + \gamma^2} \mathbf{w} \mathbf{w}^\top\right)\right),\\
        \mathcal{N}\left(q(D'), \sigma^2\left(I - \frac{\gamma^2}{\beta^2 + \gamma^2} \mathbf{w} \mathbf{w}^\top\right)\right),
    \end{align} 
    respectively. Define 
    \begin{equation}
        q'(D) := \left(I + \left(\frac{\sqrt{\beta^2 + \gamma^2}}{\beta} - 1\right) \mathbf{w} \mathbf{w}^\top\right) q(D)
    \end{equation} 
    and let $\mathcal{G}_\sigma'(D) := \mathcal{N}(q'(D), \sigma^2 I)$, such that 
    \begin{align}
        ~& \norm{q'(D) - q'(D')}^2 \\
        =& \norm{q'(D) - q'(D')}^2 + \frac{\gamma^2}{\beta^2} \left( \mathbf{w}^\top (q'(D) - q'(D')) \right)^2 \\
        \leq& \norm{q'(D) - q'(D')}^2 + \frac{\gamma^2}{\beta^2} \norm{q'(D) - q'(D')}^2.
    \end{align}

    Therefore, the sensitivity of $q'$ is at most $\frac{\sqrt{\beta^2 + \gamma^2}}{\beta} \Delta$, and by Theorem~\ref{thm:gm-dp}, $\mathcal{G}_\sigma'(\cdot)$ is $\left(\epsilon, \delta\right)$-DP for any $0 < \delta \leq 0.5$ and \begin{equation}
        \epsilon = \frac{\left(\beta^2+\gamma^2\right)\Delta^2}{2\beta^2\sigma^2} - \frac{\sqrt{\beta^2 + \gamma^2}\Delta}{\beta\sigma}\Phi^{-1}\left(\delta\right).
    \end{equation} Moreover, note that for any $D \in \mathcal{D}$,
    \begin{equation}
        \mathcal{G}_{\sigma, \mathbf{w}, \beta, \gamma}(D) \equiv \left(I - \left(1 - \frac{\beta}{\sqrt{\beta^2+\gamma^2}}\right) \mathbf{w} \mathbf{w}^\top\right) \mathcal{G}_\sigma'(D).
    \end{equation}
    Therefore, by robustness against post-processing, $\mathcal{G}_{\sigma, \mathbf{w}, \beta, \gamma}$, and hence $\M_{\sigma,\mathbf{w}, \beta,\gamma}$, satisfy the same DP guarantee.
\end{proof}
    \section{Additional Mitigation via Distributed Gaussian Mechanism}\label{app:mit-distributed}

Assume there are $N$ semi-honest servers, each possessing a local database $D_i \left(1\leq i \leq N\right)$, such that a query $q: \D \to \R^d$ can be computed by having each server $i$ compute a local query $q^{(i)}: \D_i \to \R^d$ on $D_i$ and summing the results:
\begin{equation}
    q\left(D_1\Vert D_2\Vert\dots \Vert D_N\right)=\sum_{i=1}^N q^{(i)}(D_i).
\end{equation}
Note that histogram queries and the gradient computation in DP-SGD both satisfy this assumption. To achieve DP, it is common practice to execute the \emph{distributed Gaussian mechanism}, i.e., for each server to independently perturb its local query result with Gaussian noise~\cite{DBLP:conf/eurocrypt/DworkKMMN06}:
\begin{equation}
    \mathcal{M}^{(i)}_\sigma\left(D\right):= \mathcal{N}\left(q^{(i)}\left(D_i\right), \sigma^2 I_d\right),
\end{equation}
and aggregate the results using secure aggregation (SecAgg~\cite{secagg}) based on $t$-out-of-$N$ secret sharing~\cite{DBLP:journals/cacm/Shamir79}.

We consider a scenario where \textbf{a subset of $N_b$ servers (indexed by $i_b$) has been backdoored} by the SCBP (e.g., they have downloaded the compromised software) such that they each execute $\mathbf{y}_{i_b}\sample \M^{\left(i_b\right)}_{\sigma, \mathbf{w}, \beta, \gamma}\left(D_{i_b}\right)$, while another subset of $N_c$ servers (disjoint from the backdoored servers, indexed by $i_c$) colludes with the SCBP and the BPA and shares knowledge of the backdoor key $\mathbf{w}$. WLOG, we assume that $1\leq i_b \leq N_b < N - N_c +1 \leq i_c \leq N$ to simplify the discussion. We analyze the privacy loss of each backdoored server $i_b$. For simplicity, we represent the level of privacy protection by the total scale of all added Gaussian noise, which translates to the privacy guarantees defined in Theorems~\ref{thm:gm-dp} and \ref{thm:gpm-covert}.

If $N_c < t$, the $N_c$ colluding servers gain no information about the other servers' outputs except for their sum. Specifically, their view is equivalent to the procedure outlined in Figure~\ref{fig:colluding-server-view}, which is further equivalent to
\begin{equation}
    \mathcal{N}\left(\sum_{i=1}^{N-N_c}q^{(i)}\left(D_i\right), \left(N-N_b-N_c\right)\sigma^2 I_d\right)
\end{equation}
post-processed with the hCLWE noise in Line~\ref{pcln:distributed-adv-view-hclwe}. In other words, the effective noise scale for privacy protection is $\sqrt{N-N_b-N_c}\sigma$, compared with $\sqrt{N-N_b}\sigma$ when no backdoor is installed.

\begin{figure}
    \centering
    \procedureblock[space=auto, linenumbering]{$\sf{View}\left(D_1\Vert D_2\Vert \dots \Vert D_{N-N_c}; \sf{pp}\right)$}{
        \mathbf{y}\gets \sum_{i=1}^{N-N_c}q^{(i)}\left(D_i\right) \\
        \pcfor N_b + 1\leq i \leq N-N_c \pcdo\\
            \mathbf{y}\sample \mathbf{y} + \mathcal{N}\left(\mathbf{0}, \sigma^2 I_d\right)\\
        \pcendfor\\
        \pcfor 1\leq i_b \leq N_b \pcdo\\
            \label{pcln:distributed-adv-view-hclwe} \mathbf{y}\sample \mathbf{y} + \sqrt{2\pi}\sigma\mathcal{H}_{\mathbf{w}, \beta, \gamma}\\
        \pcendfor\\
        \pcreturn \mathbf{y}
    }
    \caption{The view of $N_c$ colluding servers when $N_c < t$.}
    \label{fig:colluding-server-view}
\end{figure}

However, when $N_c\geq t$, the $N_c$ colluding servers can fully recover each $\mathbf{y}_{i_b}$. With additional knowledge of $\mathbf{w}$, the compromised privacy guarantees analyzed in Section~\ref{sec:privacy} still apply, and the distinguishing attacks discussed in Section~\ref{sec:da} remain possible. Therefore, a sufficient number of non-backdoored and non-colluding servers is necessary to ensure effective privacy guarantees.

\subsection{Reduction to Central DP} The mitigation against the backdoor can also be adapted to the central DP case. Specifically, alongside the backdoored server $i_b=1$ executing the GPM, we introduce another server $i=2$ that is \textbf{1)} not backdoored by the SCBP, and \textbf{2)} not colluding with either the SCBP or the BPA. This server is only responsible for adding a copy of (non-compromised) Gaussian noise. Thus, a potentially backdoored server executing a central DP mechanism may rely on an external, semi-honest server to ensure the privacy of its query results.

The two servers execute the mechanism as follows, which is equivalent to the distributed scenario with $N=2$, $N_b=1$, and $N_c=0$.

\begin{itemize}[leftmargin=*]
    \item Server 1 computes $\mathbf{y}_1\sample \M_{\sigma, \mathbf{w}, \beta, \gamma}\left(D_1\right)$ and transmits $\mathbf{y}_1$ to Server 2.
    \item Server 2 computes and releases $\mathbf{y}\sample \mathbf{y}_1 + \mathcal{N}\left(\mathbf{0}, \sigma^2I_d\right)$.
\end{itemize}

Since Server 2 does not collude with the SCBP or the BPA, its view of $D_1$ maintains the same level of privacy specified in Theorems~\ref{thm:gm-dp} and \ref{thm:gpm-covert}. On the other hand, the BPA's view, despite knowing $\mathbf{w}$, remains statistically DP due to the additional noise added by Server 2. However, the average $\ell^2$ error increases by a factor of $\sqrt{2}$ due to the added noise.
    \section{Additional Experimental Results on Varying \texorpdfstring{$\gamma$}{gamma} in Table~\ref{tab:cifar10-da}}

\label{app:additional-exp-gamma}

\begin{table}[!t]
    \centering
    \caption{Distinguishing attack success rates on DP-SGD with CIFAR-10 using larger values of $\gamma \in \left\{20\sqrt{d}, 200\sqrt{d}\right\}$. All instances achieve success rates of $\geq 99\%$.}
    \label{tab:cifar10-da-addl}
    \autofit{\begin{tabular}{lccccccccccccccccc}
\hline
\multirow{2}{*}{Model}       & \multirow{2}{*}{$d$}  & \multirow{2}{*}{Converged}          & \multirow{2}{*}{$\epsilon^*$} & \multirow{2}{*}{$\gamma$}         & \multicolumn{3}{c}{$\beta\left(C=4\right)$} & \multicolumn{3}{c}{$\beta\left(C=8\right)$} & \multirow{2}{*}{$\gamma$}         & \multicolumn{3}{c}{$\beta\left(C=4\right)$} & \multicolumn{3}{c}{$\beta\left(C=8\right)$} \\
                             &                       &                                     &                               &                                   & $10^{-5}$     & $10^{-4}$    & $10^{-3}$    & $10^{-5}$     & $10^{-4}$    & $10^{-3}$    &                                   & $10^{-5}$     & $10^{-4}$    & $10^{-3}$    & $10^{-5}$     & $10^{-4}$    & $10^{-3}$    \\ \hline
\multirow{8}{*}{VGG19}       & \multirow{8}{*}{39M}  & \multirow{4}{*}{\ding{55}}          & 0.125                         & \multirow{8}{*}{$1.2\times 10^5$} & 100           & 100          & 99           & 100           & 100          & 99           & \multirow{8}{*}{$1.2\times 10^6$} & 100           & 100          & 99           & 100           & 100          & 99           \\
                             &                       &                                     & 0.25                          &                                   & 100           & 100          & 100          & 100           & 100          & 99           &                                   & 100           & 100          & 100          & 100           & 100          & 100          \\
                             &                       &                                     & 0.5                           &                                   & 100           & 100          & 99           & 100           & 100          & 100          &                                   & 100           & 100          & 100          & 100           & 100          & 100          \\
                             &                       &                                     & 1                             &                                   & 100           & 100          & 100          & 100           & 100          & 100          &                                   & 100           & 100          & 100          & 100           & 100          & 100          \\
                             &                       & \multirow{4}{*}{\ding{52} (94.0\%)} & 0.125                         &                                   & 100           & 100          & 100          & 100           & 100          & 99           &                                   & 100           & 100          & 100          & 100           & 100          & 100          \\
                             &                       &                                     & 0.25                          &                                   & 100           & 100          & 99           & 100           & 100          & 99           &                                   & 100           & 100          & 100          & 100           & 100          & 100          \\
                             &                       &                                     & 0.5                           &                                   & 100           & 100          & 100          & 100           & 100          & 100          &                                   & 100           & 100          & 100          & 100           & 100          & 100          \\
                             &                       &                                     & 1                             &                                   & 100           & 100          & 100          & 100           & 100          & 100          &                                   & 100           & 100          & 100          & 100           & 100          & 100          \\ \hline
\multirow{8}{*}{ResNet50}    & \multirow{8}{*}{23M}  & \multirow{4}{*}{\ding{55}}          & 0.125                         & \multirow{8}{*}{$9.7\times 10^4$} & 100           & 100          & 99           & 100           & 100          & 100          & \multirow{8}{*}{$9.7\times 10^5$} & 100           & 100          & 100          & 100           & 100          & 100          \\
                             &                       &                                     & 0.25                          &                                   & 100           & 100          & 100          & 100           & 100          & 99           &                                   & 100           & 100          & 100          & 100           & 100          & 100          \\
                             &                       &                                     & 0.5                           &                                   & 100           & 100          & 100          & 100           & 100          & 100          &                                   & 100           & 100          & 99           & 100           & 100          & 100          \\
                             &                       &                                     & 1                             &                                   & 100           & 100          & 100          & 100           & 100          & 100          &                                   & 100           & 100          & 100          & 100           & 100          & 100          \\
                             &                       & \multirow{4}{*}{\ding{52} (93.7\%)} & 0.125                         &                                   & 100           & 100          & 100          & 100           & 100          & 100          &                                   & 100           & 100          & 100          & 100           & 100          & 100          \\
                             &                       &                                     & 0.25                          &                                   & 100           & 100          & 99           & 100           & 100          & 99           &                                   & 100           & 100          & 100          & 100           & 100          & 100          \\
                             &                       &                                     & 0.5                           &                                   & 100           & 100          & 100          & 100           & 100          & 100          &                                   & 100           & 100          & 100          & 100           & 100          & 100          \\
                             &                       &                                     & 1                             &                                   & 100           & 100          & 100          & 100           & 100          & 100          &                                   & 100           & 100          & 100          & 100           & 100          & 100          \\ \hline
\multirow{8}{*}{MobileNetV2} & \multirow{8}{*}{2.2M} & \multirow{4}{*}{\ding{55}}          & 0.125                         & \multirow{8}{*}{$3.0\times 10^4$} & 100           & 100          & 99           & 100           & 100          & 100          & \multirow{8}{*}{$3.0\times 10^5$} & 100           & 100          & 100          & 100           & 100          & 100          \\
                             &                       &                                     & 0.25                          &                                   & 100           & 100          & 100          & 100           & 100          & 100          &                                   & 100           & 100          & 100          & 100           & 100          & 100          \\
                             &                       &                                     & 0.5                           &                                   & 100           & 100          & 100          & 100           & 100          & 99           &                                   & 100           & 100          & 100          & 100           & 100          & 100          \\
                             &                       &                                     & 1                             &                                   & 100           & 100          & 100          & 100           & 100          & 99           &                                   & 100           & 100          & 100          & 100           & 100          & 100          \\
                             &                       & \multirow{4}{*}{\ding{52} (93.9\%)} & 0.125                         &                                   & 100           & 100          & 100          & 100           & 100          & 100          &                                   & 100           & 100          & 100          & 100           & 100          & 100          \\
                             &                       &                                     & 0.25                          &                                   & 100           & 100          & 100          & 100           & 100          & 100          &                                   & 100           & 100          & 100          & 100           & 100          & 100          \\
                             &                       &                                     & 0.5                           &                                   & 100           & 100          & 100          & 100           & 100          & 100          &                                   & 100           & 100          & 100          & 100           & 100          & 100          \\
                             &                       &                                     & 1                             &                                   & 100           & 100          & 100          & 100           & 100          & 100          &                                   & 100           & 100          & 100          & 100           & 100          & 100          \\ \hline
\end{tabular}}
\end{table}

Table~\ref{tab:cifar10-da-addl} details the success rates of the distinguishing attack on DP-SGD with CIFAR-10 using larger values of $\gamma$, specifically $\gamma = 20\sqrt{d}$ and $\gamma = 200\sqrt{d}$. In contrast to the results using the default $\gamma = 2\sqrt{d}$ (Table~\ref{tab:cifar10-da}), these larger values further increase the success rates to at least 99\% across all tested instances.
\fi

\end{document}